\documentclass[pdftex, 11pt, a4paper]{article}
\usepackage{amsmath, amssymb, amsfonts, verbatim, graphicx,
  rotating, multirow, titlesec, natbib, caption2, a4wide ,amsthm, subfigure, bbm}
\usepackage{algorithm}
\usepackage{algorithmic}
\titleformat{\section}
  {\normalfont\Large\bfseries\centering}{\thesection}{1em}{}
\bibpunct{(}{)}{;}{a}{}{,}

\newcommand{\mat}[1]{{\boldsymbol #1}}
\newcommand{\vect}[1]{{\boldsymbol #1}}
\newcommand{\titleq}[1]{{#1}}
\newcommand{\expo}[1]{{\exp\left({#1}\right)}}
\newcommand{\func}[1]{{\tt#1\rm}}
\newcommand{\tint}[0]{u}
\newcommand{\imag}[0]{\imath\cdot}
\newcommand{\imagND}[0]{\imath}

\newtheorem{proposition}{Proposition}

\graphicspath{{/u/sigrist/R/Precipitation/PostprocFreq/FiguresPaper/}{/Users/fabiosigrist/Documents/Work/R/Precipitation/PostprocFreq/FiguresPaper/}}

\begin{document}
\title{Stochastic partial differential equation based modelling of large space-time data sets}
\author{Fabio Sigrist, Hans R. K\"unsch, Werner A. Stahel\\ Seminar for
  Statistics, Department of Mathematics, ETH Z\"urich\\ 8092 Z\"urich, Switzerland}
  \date{March 24, 2014}
\maketitle


\begin{center}
\textbf{Abstract}
\end{center}

Increasingly larger data sets of processes in space and time ask for statistical
models and methods that can cope with such data. We show that the solution of a stochastic advection-diffusion partial differential equation provides a flexible model class for spatio-temporal
processes which is computationally feasible also
for large data sets. The Gaussian process defined through the stochastic partial differential equation has in general a nonseparable
covariance structure. Furthermore, its parameters can be
physically interpreted as explicitly modeling phenomena such as transport and diffusion that occur in
many natural processes in diverse fields ranging from environmental sciences to ecology. In order to obtain computationally efficient 
statistical algorithms we use spectral methods to solve the stochastic partial differential equation. This has the advantage that approximation errors do not accumulate
over time, and that in the spectral space the computational cost grows linearly with the dimension, the total computational 
costs of Bayesian or frequentist 
inference being dominated by the fast Fourier transform. The proposed model is
applied to postprocessing of precipitation forecasts from a numerical
weather prediction model for northern Switzerland. In contrast to the
raw forecasts from the numerical model, the postprocessed 
forecasts are calibrated and quantify prediction uncertainty. Moreover, they outperform the raw forecasts, in the sense that they have a lower mean absolute error.
\vspace*{.3in}

\newpage

\section{\titleq{Introduction}}
Space-time data arise in many applications, see \cite{CrWi11}
for an introduction and an overview. Increasingly larger space-time data sets are obtained, for
instance, from remote sensing satellites or deterministic physical models
such as numerical weather prediction (NWP) models. Statistical
models are needed that can cope with such data.
 
As \cite{WiHo10} point out, there are two basic paradigms for constructing spatio-temporal models. The
first approach is descriptive and follows the traditional geostatistical
paradigm, using joint space-time covariance functions \citep{CrHu99, Gn02,
  Ma03, Wi03, St05, PaSc06}. The second approach is dynamic and combines
ideas from time-series and spatial statistics \citep{SoSw96, WiCr99, HuHs04, XuWiFo05, GeEtAl05,
  JoCrHu07, SiKuSt11}.

Even for purely spatial data, developing methodology which can handle large
data sets is an active area of research. \cite{BaBrGe04} refer to this as the
``big n problem''. Factorizing large covariance matrices is not
possible without assuming a special structure or using approximate
methods. Using low rank matrices is one approach \citep{NyWiRo02, BaEtAl08,
  CrJo08, St08, Wi2010}. Other proposals include using Gaussian Markov
random-fields (GMRF) \citep{RuTj02, RuHe05, LiLiRu10} or applying tapering
\citep{FuGeNy06} thereby obtaining sparse precision or covariance
matrices, respectively, for which calculations can be done efficiently.  Another proposed solution is to approximate the likelihood so that it can be
evaluated faster \citep{Ve88, StChWe04, Fu07, EiEtAl11}. \cite{RoWi05} and
\cite{Pa07} use Fourier functions to reduce computational costs.

In a space-time setting, the situation is the same, if not worse: one runs
into a computational bottleneck with high dimensional data since the computational cost to factorize dense $NT\times NT$ covariance matrices is $O((NT)^3)$, $N$ and $T$ being the number of points in space and time,
respectively. Moreover, specifying flexible and realistic space-time
covariance functions is a nontrivial task.

In this paper, we follow the dynamic approach and study models which are
defined through a stochastic advection-diffusion partial differential equation (SPDE). This has
the advantage of providing physically motivated parametrizations of
space-time covariances. We show that when solving the SPDE using Fourier functions, one
can do computationally efficient statistical inference. In the spectral
space, computational costs for the Kalman filter and backward sampling
algorithms are of order $O(NT)$. As we show, roughly speaking, this computational efficiency is due to the temporal Markov
property, the fact that Fourier functions are eigenfunctions of
the spatial differential operators, and the use of some matrix identities. The overall
computational costs are then determined by the ones of the fast Fourier
transform (FFT) \citep{CoTu65} which are $O(TN\log N)$. In addition, computational time
can be further reduced by running the $T$ different FFTs in parallel.

Defining Gaussian processes through stochastic differential equations has a
long history in statistics going back to early works such as \cite{Wh54},
\cite{He55}, and \cite{Wh62}. Later works include \cite{JoZh97} and \cite{BrKa00}. Recently,
\cite{LiLiRu10} have shown how a certain class of SPDEs can be solved using
finite elements to obtain parametrizations of spatial GMRF. Note that a potential caveat of these SPDE approaches is that it is nontrivial to generalize the linear equation to non-linear ones. 

Spectral methods for solving partial differential equations are well
established in the numerical mathematics community (see, e.g.,
\cite{GoOr77}, \cite{Fo92}, or \cite{Ha04}). In contrast, statistical models
have different requirements and goals, since the (hyper-)parameters of an (S)PDE are
not known \textit{a priori} and need to be estimated. Spectral methods have also
been used in spatio-temporal statistics, mostly for approximating or
solving deterministic integro-difference equations (IDEs) or PDEs. \cite{WiCr99} introduce a dynamic
spatio-temporal model obtained from an IDE that is approximated using a
reduced-dimensional spectral basis. Extending this work, \cite{Wi02} and
\cite{XuWiFo05} propose parametrizations of spatio-temporal processes based on
IDEs. Modeling tropical ocean surface winds, \cite{WiEtAl01} present a
physics based model based on the shallow-water equations. \citet[Chapter
7]{CrWi11} give an overview of basis function expansions in
spatio-temporal statistics.

The novel features of our work are the following. While spectral methods
have been used for approximating deterministic IDEs and PDEs in the statistical
literature, there is no article, to our knowledge, that explicitly shows how to obtain a space-time Gaussian process by solving an advection-diffusion
SPDE using the real Fourier transform. Moreover, we present computationally
efficient algorithms for doing statistical inference, which use the fast Fourier transform and the Kalman
filter. The computational burden can be additionally alleviated by applying
dimension reduction. We also give a bound on
the accuracy of the approximate solution. In the application, our main objective is to postprocess precipitation
forecasts, explicitly modeling spatial and temporal variation. The idea is that the
spatio-temporal model not only accounts for dependence, but also
captures and extrapolates dynamically an error term of the NWP model in space and time.

The remainder of this paper is organized as follows. Section \ref{ContMod}
introduces the continuous space-time Gaussian process defined through the
advection-diffusion SPDE. In Section \ref{SpecSpace}, it is shown how the solution of the SPDE can be approximated
using the two-dimensional real Fourier transform, and we give convergence
rates for the approximation. Next, in Section \ref{inference}, we show how to do computationally efficient inference. In Section \ref{Postproc}, the spatio-temporal
model is used as part of a hierarchical Bayesian model, which we then apply
for postprocessing of precipitation forecasts.

All the methodology presented in this article is implemented in the R
package \func{spate} (see \cite{SiKuSt12b}).

\section{\titleq{A Continuous Space-Time Model: The Advection-Diffusion SPDE}}\label{ContMod}
In one dimension, a fundamental process is the Ornstein-Uhlenbeck process
which is governed by a relatively simple stochastic differential equation (SDE). The process has an exponential covariance function and its
discretized version is the famous AR(1) model. In the two dimensional spatial
case, \cite{Wh54} argues convincingly that the process with a Whittle
correlation function is an ``elementary'' process (see Section
\ref{InnoProc} for further discussion). If the time dimension is added, we
think that the process defined through the stochastic partial
differential equation (SPDE) in \eqref{SPDE} has properties that make it a good candidate for an
``elementary'' spatio-temporal process. It is a linear equation that 
explicitly models phenomena such as transport and diffusion that occur in
many natural processes ranging from environmental sciences to
ecology. This means that, if desired, the parameters can be given a
physical interpretation. Furthermore, if some parameters equal zero (no advection
and no diffusion), the covariance structure reduces to a separable one with
an AR(1) structure over time and a certain covariance structure over space. 

The advection-diffusion SPDE, also called transport-diffusion SPDE, is given by
\begin{equation}\label{SPDE}
\frac{\partial}{\partial
  t}\xi(t,\vect{s})=-\vect{\mu}^T\nabla
\xi(t,\vect{s})+\nabla\cdot\mat{\Sigma}\nabla\xi(t,\vect{s})-\zeta \xi(t,\vect{s})+\epsilon(t,\vect{s}),
\end{equation}
with $\vect{s}=(x,y)^T\in \mathbb{R}^{2}$, where  $\nabla =\left(\frac{\partial }{\partial x},\frac{\partial }{\partial
    y}\right)^T$ is the gradient operator, and, for a vector field $\vect{F}=(F^x,F^y)^T$,  $\nabla\cdot
\vect{F}=\frac{\partial F^x}{\partial x}+\frac{\partial F^y}{\partial y}$ is
the divergence operator. $\epsilon(t,\vect{s})$ is a Gaussian process that is temporally white and
spatially colored. See Section \ref{InnoProc} for a discussion on the
choice of the spatial covariance function. \cite{He55} and \cite{Wh63} introduced and analyzed SPDEs of similar form
as in \eqref{SPDE}. \cite{JoZh97} also investigated SPDE based
models. Furthermore, \cite{BrKa00} obtained such an advection-diffusion SPDE as a limit of
stochastic integro-difference equation models. Without giving any concrete
details, \cite{LiLiRu10} suggested that this SPDE can be used in connection
with their GMRF method. See also \cite{SiLiRu12} and \cite{YuEtAl12}. \cite{CaEtAl13} model particulate matter concentration in space and time with a separable covariance structure and an SPDE based spatial Gaussian Markov
random field for the innovation term. \cite{AuSi12} and \cite{HuEtAl13} use systems of SPDEs to define multivariate spatial models.

The SPDE has the following interpretation. Heuristically, an SPDE specifies what happens locally at each point in space during a small time step. The first term $\vect{\mu}^T\nabla
\xi(t,\vect{s})$ models transport effects (called advection in weather
applications), $\vect{\mu}=(\mu_x,\mu_y)^T\in \mathbb{R}^2$ being a drift or velocity vector. The second term, $\nabla\cdot\mat{\Sigma}\nabla\xi(t,\vect{s})$, is a
diffusion term that can incorporate anisotropy. If $\mat{\Sigma}$ is the
identity matrix, this term reduces to the divergence ($\nabla\cdot$) of the
gradient ($\nabla$) which is the ordinary Laplace operator
$\nabla\cdot\nabla=\Delta=\frac{\partial^2}{\partial
  x^2}+\frac{\partial^2}{\partial y^2}$. The third term $-\zeta
\xi(t,\vect{s})$, $\zeta>0$, diminishes $\xi(t,\vect{s})$ at a constant rate and thus
accounts for damping. Finally, $\epsilon(t,\vect{s})$ is a source-sink or stochastic forcing
term, also called innovation term, that can be interpreted as describing, amongst others, convective
phenomena in precipitation modeling applications. 

Concerning the diffusion matrix $\mat{\Sigma}$, we suggest the following parametrization
\begin{equation}
\mat{\Sigma}^{-1}=\frac{1}{\rho_1^2}\left(\begin{matrix}\cos{\psi} &
    \sin{\psi}\\ -\gamma\cdot\sin{\psi} & \gamma\cdot\cos{\psi}\end{matrix}\right)^{T}\left(\begin{matrix}\cos{\psi} &
    \sin{\psi}\\ -\gamma\cdot\sin{\psi} & \gamma\cdot\cos{\psi}\end{matrix}\right),
\end{equation}
where $\rho_1>0$, $\gamma>0$, and $\psi\in[0,\pi/2]$. The parameters are interpreted as follows. $\rho_1$ acts as a range
parameter and controls the amount of diffusion. The parameters $\gamma$ and
$\psi$ control the amount and the direction of anisotropy. With
$\gamma=1$, isotropic diffusion is obtained.

\begin{figure}
\centering
\makebox{\includegraphics[width=\textwidth]{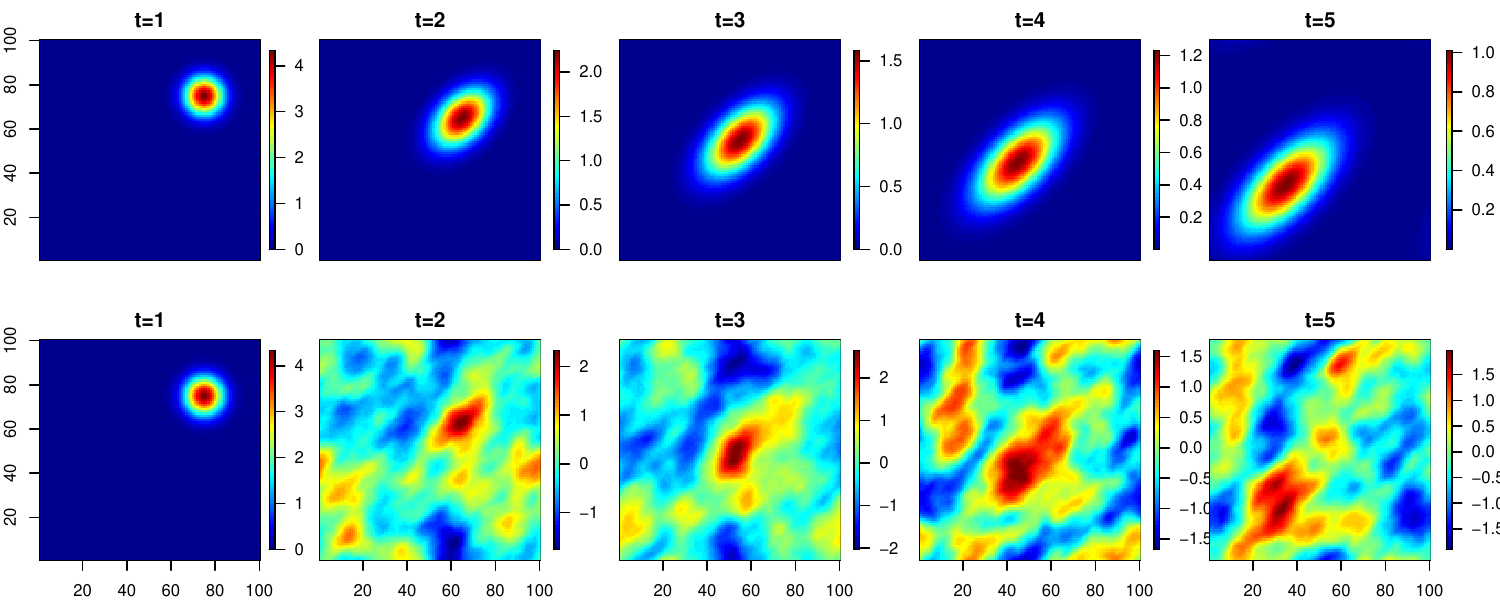}}
\caption{Illustration of the SPDE in \eqref{SPDE} and the corresponding PDE. The top row
  illustrates a solution to the PDE which corresponds to the deterministic part of the SPDE without stochastic term
  $\epsilon(t,\vect{s})$. The bottom row shows one
  sample from the distribution specified by the SPDE with a fixed initial condition. The drift vector points from
  north-east to south-west and the diffusive part exhibits anisotropy in
  the same direction. The same parameters are used
for both the PDE and the SPDE: $\zeta =
-\log(0.99), \rho_1 = 0.06, \gamma = 3, \psi = \pi/4, \mu_x = -0.1, \mu_y
= -0.1$, and for the stochastic innovations: $\rho_0 = 0.05, \sigma^2 = 0.7^2$. The color scales are different in different panels.} 
\label{fig:PDEIllus}
\end{figure}
Figure \ref{fig:PDEIllus} illustrates the SPDE in \eqref{SPDE} and the
corresponding PDE without the stochastic innovation term. The top row
  shows a solution to the PDE which corresponds the deterministic part of the SPDE that is obtained when there is no stochastic term $\epsilon(t,\vect{s})$. The
figure shows how the initial state in the top-left plot gets propagated
forward in time. The drift vector points from north-east to south-west and
the diffusive part exhibits anisotropy in the same direction. A $100 \times
100$ grid is used and the PDE is solved in the spectral domain using the
method described below in Section \ref{SpecSpace}. There is a fundamental difference
between the deterministic PDE and the probabilistic SPDE. In the first
case, a deterministic process is modeled directly. In the second case, the
SPDE defines a stochastic process. Since the operator is linear and the
input Gaussian, this process is a Gaussian process whose covariance
function is implicitly defined by the SPDE. The bottom row of Figure \ref{fig:PDEIllus} shows one
  sample from this Gaussian process. The same
initial state as in the deterministic example is used, i.e., we use a fixed
initial state. Except for the stochastic part, the same parameters are used
for both the PDE and the SPDE. For the innovations $\epsilon(t,\vect{s})$, we choose a Gaussian
process that is temporally independent and spatially structured according to the
Mat\'ern covariance function with smoothness parameter $1$. Again, the drift vector points from north-east to south-west and the diffusive part exhibits anisotropy in the same direction.

Note that the use of this spatio-temporal Gaussian process is not
restricted to situations where it is a priori known that phenomena such as
transport and diffusion occur. In the one dimensional case, it is common to use the AR(1) process in
situations where it is not a priori clear whether the modeled process
follows the dynamic of the Ornstein-Uhlenbeck SDE. In two dimensions, the same holds true
for the process with the Whittle covariance function, and even more so for the
process having an exponential covariance structure. Having this in mind, even though the SPDE in \eqref{SPDE} is physically motivated, it
can be used as a general spatio-temporal model. As the case may be, the interpretation of the
parameters can be more or less straightforward.

\subsection{Spectral Density and Covariance Function}
As can be shown using the Fourier transform (see, e.g., \cite{Wh63}), if
the innovation process $\epsilon(t,\vect{s})$ is stationary with spectral
density $\widetilde{f}(\vect{k})$, the spectrum of the stationary solution $\xi(t,\vect{s})$  of the SPDE
\eqref{SPDE} is
\begin{equation}\label{SPDESpec}
f(\omega,\vect{k})=\widetilde{f}(\vect{k})\frac{1}{(2\pi)}\left(\left(\vect{k}^T\mat{\Sigma}\vect{k}+\zeta\right)^2+\left(\omega+\vect{\mu}^T\vect{k}\right)^2\right)^{-1},
\end{equation}
where $\vect{k}$ and $\omega$ are spatial wavenumbers and temporal
frequencies. The covariance function $C(t,\vect{s})$ of $\xi(t,\vect{s})$
is then given by
\begin{equation}\label{SPDECov}
\begin{split}
C(t,\vect{s})=&\int
f(\omega,\vect{k})\expo{\imag t\omega}\expo{\imag \vect{s}'\vect{k}}d\vect{k}
d\omega\\
=&\int\widetilde{f}(\vect{k})\frac{\expo{-\imag\vect{\mu}^T\vect{k}t-(\vect{k}^T\mat{\Sigma}\vect{k}+\zeta)|t|}}{2(\vect{k}^T\mat{\Sigma}\vect{k}+\zeta)}\expo{\imag \vect{s}'\vect{k}}d\vect{k},
\end{split}
\end{equation}
where $\imagND$ denotes the imaginary number $\imagND^2=-1$, and the integration over the temporal frequencies $\omega$ follows
from the calculation of the characteristic function of the Cauchy
distribution \citep{AbSt64}. The spatial integral above has no closed form
solution but can be computed approximately by numerical integration.

Since, in general, the spectrum does not factorize into a temporal and a
spatial component, we see that $\xi(t,\vect{s})$ has a non-separable
covariance function (see \cite{GnGeGu07} for a definition of
separability). The model reduces to a separable one, though, when there is
no advection and diffusion, i.e., when both $\vect{\mu}$ and $\mat{\Sigma}$
are zero. In this case, the covariance function is given by
$C(t,\vect{s})=\frac{1}{2\zeta}\expo{-\zeta|t|}C(\vect s)$, where $C(\vect
s)$ denotes the spatial covariance function of the innovation process.

\subsection{Specification of the Innovation Process}\label{InnoProc}
It is assumed that the innovation process is white in
 time and spatially colored. In principle, one can choose
 any spatial covariance function such that the covariance function in
 \eqref{SPDECov} is finite at zero. Note that if $\widetilde{f}(\vect{k})$ is
 integrable, then $f(\omega,\vect{k})$ is also integrable. Similarly as \cite{LiLiRu10}, we opt for the
 most commonly used covariance function in spatial statistics: the Mat\'ern covariance
 function (see \cite{HaSt93}, \cite{st99}). Since in many applications the smoothness parameter is not estimable, we further restrict ourselves to the Whittle
 covariance function. This covariance function is of the form
 $\sigma^2 d/\rho_0 K_1\left(d/\rho_0\right)$ with $d$ being the Euclidean distance between two points and
$K_1\left(d/\rho_0\right)$ being the modified Bessel function of order
$1$. It is called after \cite{Wh54} who introduced it and argued
convincingly that it ``may be regarded as the 'elementary' correlation in two dimensions, similar to the exponential in one dimension.''. It can be shown that
the stationary solution of the SPDE 
\begin{equation}\label{SPDEInnov}
\left(\nabla\cdot\nabla-\frac{1}{\rho_0^2}\right) \epsilon(t,\vect{s})=\mathcal{W}(t,\vect{s}),
\end{equation}
where $\mathcal{W}(t,\vect{s})$ is a zero mean Gaussian white noise field with
variance $\sigma^2$, has the Whittle covariance function in space. From this, it follows that the spectrum of the process $\epsilon(t,\vect{s})$ is given by
\begin{equation}\label{WhittleSpec}
\widetilde{f}(\vect{k})=\frac{\sigma^2}{(2\pi)^2}\left(\vect{k}^T\vect{k}+\frac{1}{\rho_0^2}\right)^{-2}, ~~\rho_0>0,~\sigma>0.
\end{equation}
The parameter $\sigma^2$ determines the marginal variance of $\epsilon(t,\vect{s})$,
and $\rho_0$ is a spatial range parameter.

\subsection{Relation to an Integro-Difference Equation}
Assuming discrete time steps with lag $\Delta$, \cite{BrKa00} consider the following integro-difference equation (IDE)
\begin{equation}\label{IDE}
\xi(t,\vect{s})=\expo{-\Delta\zeta}\int_{\mathbb{R}^{2}}{h(\vect{s}-\vect{s}')\xi(t-\Delta,\vect{s}')d\vect{s}'}+\epsilon(t,\vect{s}),~~\vect{s}
\in \mathbb{R}^{2},
\end{equation}
with a Gaussian redistribution kernel
\begin{equation*}
h(\vect{s}-\vect{s}')=  (2\pi)^{-1}|2\Delta\mat{\Sigma}|^{-1/2}\exp\left(-(\vect{s}-\vect{s}'-\Delta\vect{\mu})^T(2\Delta\mat{\Sigma})^{-1}(\vect{s}-\vect{s}'-\Delta\vect{\mu})/2\right),
\end{equation*}
$\epsilon(t,\vect{s})$ being temporally independent and spatially
dependent. They show that in the limit $\Delta \to 0$, the solution of the
IDE and the one of the SPDE in \eqref{SPDE} coincide. The IDE is interpreted
as follows: the convolution kernel $h(\vect{s}-\vect{s}')$ determines the weight or the
 amount of influence that a location $\vect{s}'$ at previous time
 $t-\Delta$ has on the point $\vect{s}$ at current time $t$. This IDE
 representation provides an alternative way of interpreting the SPDE model and its parameters. \cite{StFrHi02} show under which conditions a dynamic model determined by an IDE as in \eqref{IDE} can be represented using a parametric joint space-time covariance
function, and vice versa. Based on the IDE in \eqref{IDE}, \cite{SiKuSt11} construct a spatio-temporal model for irregularly spaced data and apply it to obtain short term predictions of precipitation. \cite{Wi02} and \cite{XuWiFo05} also model spatio-temporal rainfall based on IDEs.

\section{\titleq{Solution in the Spectral Space}}\label{SpecSpace}
Solutions $\xi(t,\vect{s})$ of the SPDE \eqref{SPDE} are defined in
continuous space and 
time. In practice, one needs to discretize both space and time. The
resulting vector of $NT$ space-time points is in general of large
dimension.  This makes statistical inference, be it frequentist or Bayesian,
computationally difficult to impossible. However, as we show in the
following, solving the SPDE in the
spectral space alleviates the computational burden considerably and allows
for dimension reduction, if desired.

Heuristically speaking, spectral methods \citep[Chapter 7]{GoOr77,CrWi11} approximate the solution $\xi(t,\vect{s})$ by a linear combination of
deterministic spatial functions $\phi_j(\vect{s})$ with random coefficients $\alpha_j(t)$ that
evolve dynamically over time:
\begin{equation}\label{FreqApprox}
\begin{split}
\xi^K(t,\vect{s})&=\sum_{j=1}^K{\alpha_j(t)\phi_j(\vect{s})}=\vect{\phi}(\vect{s})^T\vect{\alpha}(t),
\end{split}
\end{equation}
where $\vect{\phi}(\vect{s})=(\phi_1(\vect{s}),\dots,\phi_K(\vect{s}))^T$
and $\vect{\alpha}(t)=(\alpha_1(t),\dots,\alpha_K(t))^T$. To be more
specific, we use Fourier functions
\begin{equation}\label{FourFunc}
\phi_j(\vect{s})=\exp{(\imag \vect{k}_j^T \vect{s})},
\end{equation} 
where  $\vect{k}_j =(k_j^x,k_j^y)^T$ is a spatial wavenumber. 


The advantages of using Fourier functions for solving linear, deterministic
PDEs are well known, see, e.g., \citet{Pe87}. First, differentiation in the
physical space corresponds to multiplication in the spectral space. In
other words, Fourier functions are eigenfunctions of the spatial
differential operator. Instead of approximating the differential operator
in the physical space and then worrying about approximation errors, one
just has to multiply in the spectral space, and there is no approximation
error of the operator when all the basis functions are retained. In addition, one can use the FFT for efficiently
transforming from the physical to the spectral space, and vice versa. 

Proposition \ref{PropMod} shows that Fourier functions are also useful for the
stochastic PDE \eqref{SPDE}: if the initial condition and the
innovation process are in the space spanned by a finite number of
Fourier functions, then the solution of the SPDE \eqref{SPDE} remains in
this space for all times and can be given in explicit form.
\begin{proposition}\label{PropMod}
Assume that the initial state and the innovation terms are of the form
 \begin{equation}\label{SPDEAss}
\xi^K(0,\vect{s}) =\vect{\phi}(\vect{s})^T\vect{\alpha}(0), \quad
\epsilon^K(t,\vect{s}) =\vect{\phi}(\vect{s})^T\widetilde{\vect{\epsilon}}(t)
\end{equation}
where $\vect{\phi}(\vect{s})=(\phi_1(\vect{s}),\dots,\phi_K(\vect{s}))^T$,
$\phi_j(\vect{s})$ is given in \eqref{FourFunc}, $\vect{\alpha}(0)\sim
N\left(\vect 0,\textnormal{diag}\left(\widetilde{f}_0(\vect{k}_j)\right)\right)$,
$\widetilde{f}_0(\cdot)$ being a spectral density, and $\widetilde{\vect{\epsilon}}(t)$ is a $K$-dimensional Gaussian white noise
independent of $\vect{\alpha}(0)$ with
\begin{equation}\label{SPDEAssStat}\mathrm{Cov}(\widetilde{\vect{\epsilon}}(t),\widetilde{\vect{\epsilon}}(t'))
  = \delta_{t,t'}\textnormal{diag}\left(
    \widetilde{f}(\vect{k}_j)\right),\end{equation} where $\widetilde{f}(\cdot)$ is a
spectral density and $\delta_{t,t'}$ the Kronecker delta function equaling $1$ if $t=t'$ and zero otherwise. Then the process $\xi^K(t,\vect{s}) =\vect{\phi}(\vect{s})^T\vect{\alpha}(t)$,
where the components $\alpha_j(t)$ are given by
\begin{equation}\label{alpha}
\alpha_j(t)= \expo{h_j t} \alpha_j(0) + \int_0^t \expo{h_j (t-\tint)}
\widetilde{\epsilon}_j(\tint) d\tint
\end{equation}
with $h_j = -\imag \vect{\mu}^T\vect{k}_j-\vect{k}_j^T\mat{\Sigma}\vect{k}_j - \zeta$, is a solution of the SPDE in \eqref{SPDE}. For $t \rightarrow
\infty$, the influence of the initial condition $\exp(h_j t) \alpha_j(0)$
converges to zero and the process $\xi^K(t,\vect{s})$ converges
to a time stationary Gaussian process with mean zero and 
$$\mathrm{Cov}(\xi^K(t+\Delta t,\vect{s}),\xi^K(t,\vect{s}')) =  
\vect{\phi}(\vect{s})^T \textnormal{diag}\left(\frac{-\expo{h_j\Delta t}\widetilde{f}(\vect{k}_j)}{h_j + h_j^\ast}\right) \vect{\phi}(\vect{s}')^\ast,$$
where $.^\ast$ stands for complex conjugation.
\end{proposition} 
This result shows that the solution of the SPDE is exact over time, given the frequencies included. In contrast to finite
differences, one does not accumulate errors over time. This is related to the fact that there is no need for
numerical stability conditions. For statistical applications, where the
parameters are not known a priori, this is particularly useful. The approximation
error of $\xi^K(t,\vect{s})$ to the space time stationary solution of the SPDE in \eqref{SPDE} only depends on the number of spectral terms and not on the temporal
discretization, see also Proposition \ref{PropConv} below. Since
Fourier terms are global functions, stationarity in space, but not in time,
is a necessary assumption.

\begin{proof}
By \eqref{alpha}, we have
\begin{equation*}
\begin{split}
\frac{\partial}{\partial t}\xi^K(t,\vect{s})&=\sum_{j=1}^K{\dot{\alpha}_j(t)\phi_j(\vect{s})}=\sum_{j=1}^K{(h_j \alpha_j(t) + \widetilde{\epsilon}_j(t))\phi_j(\vect{s})}.
\end{split}
\end{equation*}
On the other hand, since the functions $\phi_j(\vect{s})=\exp{(\imag
  \vect{k}_j^T\vect{s})}$ are Fourier terms, differentiation in the
physical space corresponds to multiplication in the spectral space:
\begin{equation}\label{AdvecDiff}
\vect{\mu}^T\nabla\phi_j(\vect{s})=i\vect{\mu}^T\vect{k}_j\phi_j(\vect{s})
\end{equation}
and 
\begin{equation}
\nabla\cdot\mat{\Sigma}\nabla\phi_j(\vect{s})=-\vect{k}_j^T\mat{\Sigma}\vect{k}_j\phi_j(\vect{s}). 
\end{equation} 
Therefore, by the definition of $h_j$, 
\begin{equation*}
\left(-\vect{\mu}^T\nabla + \nabla\cdot\mat{\Sigma}\nabla -\zeta\right)
\sum_{j=1}^K \alpha_j(t)\phi_j(\vect{s}) =\sum_{j=1}^K
h_j \alpha_j(t) \phi_j(\vect{s}).
\end{equation*}
Together, we have 
$$\frac{\partial}{\partial t}\xi^K(t,\vect{s})= \left(-\vect{\mu}^T\nabla + \nabla\cdot\mat{\Sigma}\nabla -\zeta\right)\xi^K(t,\vect{s})+\epsilon^K(t,\vect{s})$$
which proves the first part of the proposition. Since the real part of $h_j$
is negative, $\exp(h_jt) \rightarrow 0$ for $t \rightarrow \infty$. 
Moreover, 
\begin{equation}
\begin{split} \lim\limits_{t \rightarrow \infty} \mathrm{Cov}(\alpha_j(t+\Delta t),\alpha_{j'}(t))&= \lim\limits_{t \rightarrow \infty} \expo{h_j\Delta t} \delta_{j,j'}\widetilde{f}(\vect{k}_j) \int_0^t \expo{-(h_j + h_{j'}^\ast)(t-\tint)}d\tint\\ &= -\frac{\expo{h_j\Delta t}}{h_j +
  h_j^\ast}\delta_{j,j'}\widetilde{f}(\vect{k}_j),
\end{split}
\end{equation}
and thus the last statement follows.
\end{proof}

We assume that the forcing term $\epsilon(t,.)$, the initial state $\xi(0,.)$, and consequently also the
solution $\xi(t,.)$, are stationary in space. Recall the Cram\'er
representation for a stationary field $\epsilon(t,.)$
$$\epsilon(t,\vect{s}) = \int \exp{(\imag \vect{k}^T \vect{s})} d \widetilde \epsilon_t(\vect{k})$$
where $\widetilde\epsilon_t$ has orthogonal increments 
$\mathrm{Cov}(d\widetilde\epsilon_t(\vect{k}),d\widetilde\epsilon_{t'}(\vect{l})) =  \delta_{t,t'}\delta_{\vect{k},\vect{l}} 
\widetilde{f}(\vect{k})$
and $\widetilde{f}(\cdot)$ is the spectral density of $\epsilon(t,.)$ (see, e.g., \cite{CrLe67}). This implies that we
can approximate any stationary field, in particular also the one with a
Whittle covariance function, by a finite linear combination of complex
exponentials, and the covariance of $\widetilde{\vect{\epsilon}}(t)$ is a
diagonal matrix as required in the proposition. Its entries are specified
in \eqref{WhittleSpec}. Concerning the initial state, one can use the
stationary distribution of $\xi(t,.)$. An alternative choice is to use the same
spatial distribution as for the innovations: $\widetilde{f}_0(\cdot)=\widetilde{f}(\cdot)$.

\subsection{Approximation bound}
By passing to the limit $K \rightarrow \infty$ such that both the wavenumbers
$\vect{k}_j$ cover the entire domain $\mathbb{R}^2$ and the distance
between neighboring wavenumbers goes to zero, we obtain from
\eqref{FreqApprox} the stationary (in space and time) solution with
spectral density as in \eqref{SPDESpec}. In practice, if one uses the
discrete Fourier transform (DFT), or its fast variant, the FFT, the wavenumbers are regularly spaced and the distance
between them is fixed for all $K$ (see below). This implies that the
covariance function of an approximate solution is periodic which is equivalent
to assuming a rectangular domain being wrapped around a torus. Since in
most applications, the domain is fixed anyway, this is a reasonable assumption.

Based on the above considerations, we assume, in the following, that $\vect
s \in [0,1]^2$ with periodic boundary condition, i.e., that
$[0,1]^2$ is wrapped on a
torus. In practice, to avoid spurious periodicity, we can apply
what is called ``padding''. This means that we take $\vect s \in [0,0.5]^2$ and then
embed it in $[0,1]^2$. As in the discrete Fourier transform, if we choose $\vect s \in [0,1]^2$, it follows that the spatial wavenumbers $\vect k_j$ lie on the $n\times
n$ grid given by $D_n=\{ 2\pi\cdot(i,j): -(n/2+1)\leq i,j \leq
n/2\}=\{-2\pi(n/2+1),\dots,2\pi n/2\}^2$ with $n^2 = N = K$, $n$ being an
even natural number. We then have
the following convergence result.
\begin{proposition}\label{PropConv}
When $N \rightarrow \infty$, the approximation $\xi^N(t,\vect{s})$ converges in law to the solution
$\xi(t,\vect{s})$ of the SPDE \eqref{SPDE} with $\vect s \in [0,1]^2$ wrapped on a
torus, and we have the bound
\begin{equation}
|C(t,\vect{s})-C^N(t,\vect{s})|\leq \sigma^2_{\xi}-\sigma^2_{\xi^N} ,
\end{equation}
where  $C(t,\vect{s})$ and $C^N(t,\vect{s})$ denote the covariance functions
of $\xi(t,\vect{s})$ and $\xi^N(t,\vect{s})$, respectively, and where
$\sigma^2_{\xi}=C(0,\vect{0})$ and $\sigma^2_{\xi^N}=C^N(0,\vect{0})$
denote the marginal variances of these two processes.
\end{proposition}

\begin{proof}
Similarly as in \eqref{SPDECov} and due to
$\vect{k} \in 2\pi\cdot \mathbb{Z}^2$, it follows that the covariance function
of $\xi(t,\vect{s})$ is given by
\begin{equation}
\begin{split}
C(t,\vect{s})=&\sum_{\vect{k} \in  2\pi\cdot \mathbb{Z}^2}\int
f(\omega,\vect{k})\expo{\imag t\omega}d\omega\exp(\imag\vect{s}'\vect{k}) \\
=&\sum_{\vect{k} \in  2\pi\cdot\mathbb{Z}^2}\widetilde{f}(\vect{k})\frac{-\expo{h_{\vect k}t}}{h_{\vect k}+h_{\vect k}^*}\exp(\imag\vect{s}'\vect{k}),
\end{split}
\end{equation}
where $h_{\vect k}= -\imag \vect{\mu}^T\vect{k}-\vect{k}^T\mat{\Sigma}\vect{k} - \zeta$.
From Proposition \ref{PropMod} we know that the approximate solution $\xi^N(t,\vect{s})$ has the covariance function
\begin{equation}
C^N(t,\vect{s})=\sum_{\vect{k} \in D_n}\widetilde{f}(\vect{k})\frac{-\expo{h_{\vect k}t}}{h_{\vect k}+h_{\vect k}^*}\exp(\imag\vect{s}'\vect{k}).
\end{equation}
It follows that
\begin{equation}
\begin{split}
|C(t,\vect{s})-C^N(t,\vect{s})|=&\left|\sum_{\vect{k} \in  2\pi\cdot \mathbb{Z}^2}\widetilde{f}(\vect{k})\frac{-\expo{h_{\vect
        k}t}}{h_{\vect k}+h_{\vect k}^*}(1-\mathbbmss{1}_{\{\vect k \in
    D_n\}} )\exp(\imag\vect{s}'\vect{k})\right|\\
\leq&\sum_{\vect{k} \in  2\pi\cdot \mathbb{Z}^2}\widetilde{f}(\vect{k})\frac{-1}{h_{\vect k}+h_{\vect k}^*}(1-\mathbbmss{1}_{\{\vect k \in D_n\}} )\\
=&\sigma^2_{\xi}-\sigma^2_{\xi^N}.
\end{split}
\end{equation}
\end{proof}
Not surprisingly, this result tells us that the rate of convergence
essentially depends on the smoothness properties of the process
$\xi(t,\vect{s})$, i.e., on how fast the spectrum decays. The smoother
$\xi(t,\vect{s})$, that is, the more variation is explained by low
frequencies, the faster is the convergence of the approximation.

Note that there is a conceptual difference between the stationary solution
of the SPDE \eqref{SPDE} with $\vect{s} \in \mathbb{R}^{2}$ and the
periodic one with
$\vect s \in [0,1]^2$ wrapped on a torus. For the sake of notational
simplicity, we have denoted both of them by $\xi(t,\vect{s})$. The finite
dimensional solution $\xi^N(t,\vect{s})$ is an
approximation to both of the above infinite dimensional solutions. The
above convergence result, though, only holds true for the solution
on the torus.

\subsection{Real Fourier Functions and Discretization in Time and Space}
To apply the model to real data, we have to discretize
it. In the following, we consider the process $\xi(t,\vect{s})$
on a regular grid of $n \times n =N$  spatial locations
$\vect{s}_1,\dots,\vect{s}_N$ in $[0,1]^2$ and at equidistant time points
$t_1,\dots,t_T$ with $t_i-t_{i-1}=\Delta$. Note that these two assumptions can be
easily relaxed, i.e., one can have irregular spatial observation locations and
non-equidistant time points. The former can be achieved by adopting a data
augmentation approach (see, for instance, \cite{SiKuSt11}) or by using an incidence matrix (see Section \ref{dimred}). The
latter can be done by taking a time varying $\Delta$.

For the sake of illustration, we have stated the results in the previous
section using complex Fourier functions. However, when discretizing the model, one
obtains a linear Gaussian state space model with a propagator matrix
$\mat{G}$ that contains complex
numbers, due to \eqref{AdvecDiff}. To avoid this, we replace the complex terms $\exp{(\imag \vect{k}_j^T\vect{s})}$ with real $\cos(\vect{k}_j^T\vect{s})$ and
$\sin(\vect{k}_j^T\vect{s})$ functions. In other words, we use the real instead of the
complex Fourier transform. The above results then still hold true, since for
real valued data, the real Fourier transform is equivalent to the complex
one. For notational simplicity, we will drop the superscript ``$^K$'' from
$\xi^K(t,\vect{s})$. The distinction between the approximation and
the true solution is clear from the context.

\begin{proposition}\label{PropRealDiscr}
On the above specified discretized spatial and temporal domain and using the
real Fourier transform, with initial state $\vect \alpha(t_0)\sim N(0,\widetilde{\mat{Q}}_0)$, $\widetilde{\mat{Q}}_0$ diagonal, a stationary solution of the SPDE \eqref{SPDE} is of the form
\begin{align}\label{SolFTVec}
\vect{\xi}(t_{i+1})&=\mat{\Phi}\vect{\alpha}(t_{i+1}),\\
\vect{\alpha}(t_{i+1})&=\mat{G}\vect{\alpha}(t_i)+\widetilde{\vect{\epsilon}}(t_{i+1}),~~\widetilde{\vect{\epsilon}}(t_{i+1})\sim
N(0,\widetilde{\mat{Q}}),\label{SolCoef}
\end{align}
with stacked vectors $\vect{\xi}(t_i)=(\xi(t_i,\vect s_1),\dots,\xi(t_i,\vect
s_N))^T$ and cosine and sine coefficients
$\vect{\alpha}(t_i)=\left(\alpha_1^{(c)}(t_i),\dots,\alpha_4^{(c)}(t_i),\alpha_5^{(c)}(t_i),\alpha_5^{(s)}(t_i),\dots,\alpha_{K/2+2}^{(c)}(t_i),\alpha_{K/2+2}^{(s)}(t_i)\right)^T,$
where $\mat{\Phi}$ applies the discrete, real Fourier transformation,
$\mat{G}$ is a block diagonal matrix with $2 \times 2$ blocks, and
$\widetilde{\mat{Q}}$ is a diagonal matrix. The above matrices are defined
as follows.
\begin{itemize}
\item
  $\mat{\Phi}=\left[\vect{\phi}(\vect{s}_1),\dots,\vect{\phi}(\vect{s}_N)\right]^T,$\\
  $\vect{\phi}(\vect{s}_l)=\left(\phi_1^{(c)}(\vect{s}_l),\dots,\phi_4^{(c)}(\vect{s}_l),\phi_5^{(c)}(\vect{s}_l),\phi_5^{(s)}(\vect{s}_l),\dots,\phi_{K/2+2}^{(c)}(\vect{s}_l),\phi_{K/2+2}^{(s)}(\vect{s}_l)\right)^T,$\\
  $\phi^{(c)}_j(\vect{s}_l)=\cos(\vect{k}_j^T\vect{s}_l),~~\phi^{(s)}_j(\vect{s}_l)=\sin(\vect{k}_j^T\vect{s}_l)$, $\l=1,\dots,n^2$

\item  $[\mat G]_{1:4,1:4}=\textnormal{diag}\left(\expo{ -\Delta
    (\vect{k}_j^T\mat{\Sigma}\vect{k}_j+ \zeta)}\right),$

$[\mat  G]_{5:K,5:K}=\textnormal{diag}\left(\expo{-\Delta(\vect{k}_j^T\mat{\Sigma}\vect{k}_j +
    \zeta)}\left(\cos(\Delta\vect{\mu}^T\vect{k}_j)\mat 1_2-\sin(\Delta\vect{\mu}^T\vect{k}_j)\mat{J}_2\right)\right),$
where 
\begin{equation}
\mat 1_2=\left(\begin{matrix} 1&0 \\0  &1 \end{matrix}\right), ~~\mat{J}_2=\left(\begin{matrix} 0&1 \\-1  &0 \end{matrix}\right),
\end{equation}
\item
  $\widetilde{\mat{Q}}=\text{\textnormal{diag}}\left(\widetilde{f}(\vect{k}_j)\frac{1-\expo{-2\Delta(\vect{k}_j^T\mat{\Sigma}\vect{k}_j
        + \zeta)}}{2(\vect{k}_j^T\mat{\Sigma}\vect{k}_j +
      \zeta)}\right),$
      
      \item $\widetilde{\mat{Q}}_0=(\mat 1_N-\mat G \mat G^T)^{-1}\widetilde{\mat{Q}}$.
\end{itemize}
\end{proposition}

In summary, at each time point $t$ and spatial point $\vect s_l$, $l=1,\dots,n^2$, the solution $\xi(t,\vect{s}_l)$ is
the discrete real Fourier transform of the random coefficients $\vect{\alpha}(t)$
\begin{equation}\label{SolFT}
  \begin{split}
\xi(t,\vect{s}_l) &=
\sum_{j=1}^{4}\alpha^{(c)}_j(t)\phi^{(c)}_j(\vect{s}_l)+\sum_{j=5}^{K/2+2}{\left(\alpha^{(c)}_j(t)\phi^{(c)}_j(\vect{s}_l)+\alpha^{(s)}_j(t)\phi^{(s)}_j(\vect{s}_l)\right)}\\
&=\vect{\phi}(\vect{s}_l)^T\vect{\alpha}(t),
\end{split} 
\end{equation}
and the Fourier coefficients $\vect{\alpha}(t)$ evolve dynamically over
time according to the vector autoregression in \eqref{SolCoef}. The first four terms are cosine terms and, afterward, there are
cosine - sine pairs. This is a peculiarity of the real Fourier
transform. It is due to the fact that for four wavenumbers $\vect
k_j$, the
sine terms equal zero on the grid, i.e.,
$\sin(\vect k_j^T\vect s_l)=0$, for all $l=1,\dots,n^2$ and $\vect
k_j\in \{(0,0)^T,(0,n \pi)^T,(n\pi,0)^T,(n\pi,n\pi)^T\}$ (see Figure \ref{fig:WaveIllus}). The above equations \eqref{SolFTVec} and \eqref{SolCoef} form a linear Gaussian state space model with
parametric propagator matrix $\mat{G}$ and innovation covariance matrix
$\widetilde{\mat{Q}}$, the parametrization being determined by the
corresponding SPDE.

The model in \eqref{SolFTVec} and \eqref{SolCoef} is similar to the one
discussed in \citet[Chapter 7]{CrWi11}, but the derivation as an exact solution
to the stochastic PDE \eqref{SPDE} rather than a deterministic PDE is different.

\begin{proof}
Similarly as in Proposition \ref{PropMod}, we first derive the continuous time
solution. Using 
\begin{equation*}
\vect{\mu}^T\nabla\phi^{(c)}_j(\vect{s}_l)=-\vect{\mu}^T\vect{k}_j\phi^{(s)}_j(\vect{s}_l),~~ \vect{\mu}^T\nabla\phi^{(s)}_j(\vect{s}_l)=\vect{\mu}^T\vect{k}_j\phi^{(c)}_j(\vect{s}_l),
\end{equation*}
\begin{equation*}
\nabla\cdot\mat{\Sigma}\nabla\phi^{(c)}_j(\vect{s}_l)=-\vect{k}_j^T\mat{\Sigma}\vect{k}_j\phi^{(c)}_j(\vect{s}_l),~~ \nabla\cdot\mat{\Sigma}\nabla\phi^{(s)}_j(\vect{s}_l)=-\vect{k}_j^T\mat{\Sigma}\vect{k}_j\phi^{(s)}_j(\vect{s}_l),
\end{equation*} 
and the same arguments as in the proof of Proposition \ref{PropMod}, it
follows that the continuous time solution is of the form \eqref{SolFT}. For
each pair of cosine - sine coefficients $\vect
\alpha_j(t)=(\alpha^{(c)}_j(t),\alpha^{(s)}_j(t))^T$ we have
\begin{equation}\label{RealSol}
\vect \alpha_j(t)=e^{\mat{H}_j t}\vect
\alpha_j(0)+\int_0^te^{\mat{H}_j (t-\tint)} \widetilde{\vect{\epsilon}}_j(\tint)d\tint,
\end{equation}
where
\begin{equation*}
\mat{H}_j=\left(\begin{matrix} -\vect{k}_j^T\mat{\Sigma}\vect{k}_j - \zeta&-\vect{\mu}^T\vect{k}_j \\\vect{\mu}^T\vect{k}_j  &-\vect{k}_j^T\mat{\Sigma}\vect{k}_j - \zeta \end{matrix}\right).
\end{equation*}
Now $\mat{H}_j$ can be written as
\begin{equation*}
\mat{H}_j=(-\vect{k}_j^T\mat{\Sigma}\vect{k}_j - \zeta) \mat  1_2 -
\vect{\mu}^T\vect{k}_j  \mat{J_2}, 
\end{equation*}
where 
\begin{equation*}
\mat 1_2=\left(\begin{matrix} 1&0 \\0  &1 \end{matrix}\right), ~~\mat{J}_2=\left(\begin{matrix} 0&1 \\-1  &0 \end{matrix}\right).
\end{equation*}
Since $\mat 1_2$ and $\mat{J}_2$ commute, we have
\begin{equation}\label{propPairs}
\begin{split}
e^{\mat{H}_j t}=&\expo{-t(\vect{k}_j^T\mat{\Sigma}\vect{k}_j+ \zeta)\mat 1_2}
\expo{-t\vect{\mu}^T\vect{k}_j\mat{J}_2}\\
=&\expo{-t(\vect{k}_j^T\mat{\Sigma}\vect{k}_j + \zeta)}\left(\cos(t\vect{\mu}^T\vect{k}_j)\mat 1_2-\sin(t\vect{\mu}^T\vect{k}_j)\mat{J}_2\right).
\end{split}
\end{equation}
For the calculation of the exponential function of the matrix $\mat{J}_2$,
see, e.g., \citet[Chapter 4]{Br07}.

Analogously, one derives for the first four cosine terms
\begin{equation}\label{propCosine}
 \alpha_j^c(t)=e^{-\left(\vect{k}_j^T\mat{\Sigma}\vect{k}_j + \zeta\right) t}
\alpha_j^c(0)+\int_0^te^{-\left(\vect{k}_j^T\mat{\Sigma}\vect{k}_j +
    \zeta\right)(t-\tint)} \widetilde{\epsilon}_j(\tint)d\tint, ~~j=1,\dots 4.
\end{equation}
The above expression \eqref{propPairs} and \eqref{propCosine} give the
propagator matrix $\mat G$.

For the discrete time solution, in addition to the propagation $$\vect \alpha_j(t+\Delta)=e^{\mat{H}_j \Delta}\vect
\alpha_j(t),$$ we need to calculate the covariance of the integrated
stochastic innovation term $$\int_t^{t+\Delta}e^{\mat{H}_j
  (t+\Delta-\tint)} \widetilde{\vect{\epsilon}}_j(\tint)d\tint .$$
 This is calculated as
\begin{equation*}
\begin{split}
\int_t^{t+\Delta}e^{\mat{H}_j (t+\Delta-\tint)} \widetilde{f}(\vect{k}_j)e^{\mat{H}'_j
  (t+\Delta-\tint)}d\tint&=\int_0^{\Delta}e^{\mat{H}_j (\Delta-\tint)} \widetilde{f}(\vect{k}_j)e^{\mat{H}'_j
  (\Delta-\tint)}d\tint\\
&=\int_0^{\Delta}\widetilde{f}(\vect{k}_j)\expo{-2(\vect{k}_j^T\mat{\Sigma}\vect{k}_j+ \zeta)(\Delta-\tint)}\mat 1_2d\tint\\
&=\widetilde{f}(\vect{k}_j)\frac{1-\expo{-2(\vect{k}_j^T\mat{\Sigma}\vect{k}_j+ \zeta)\Delta}}{2(\vect{k}_j^T\mat{\Sigma}\vect{k}_j +
    \zeta)}\mat 1_2.
\end{split}
\end{equation*}
For the first four cosine terms, calculations are done analogously. The covariance matrix $\widetilde{\mat{Q}}_0$ of the initial state $\vect \alpha(t_0)$ is assumed to be the covariance matrix of the stationary distribution of $\vect \alpha(t_i)$. Note that  $\widetilde{\mat{Q}}_0$  is diagonal since $\mat G \mat G^T$ is diagonal, see the proof of Algorithm \ref{skf} in Section \ref{kfbsss}. This then gives the result in \eqref{SolFTVec} and \eqref{SolCoef}.
\end{proof}

\begin{figure}[!ht]
\begin{center}
\includegraphics[width=0.5\textwidth]{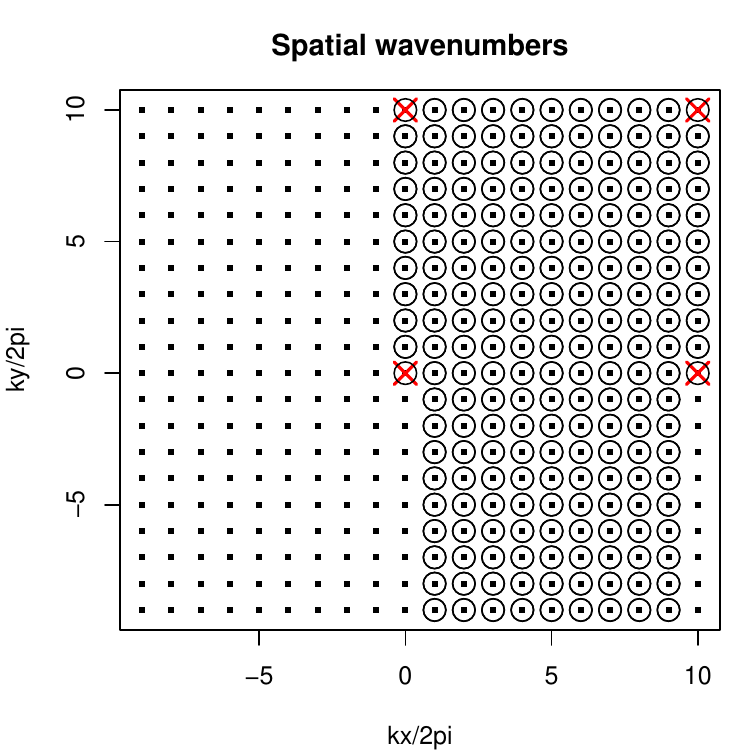} 
\end{center}
\caption{Illustration of spatial wavenumbers for the two-dimensional
  discrete real Fourier transform with $n^2=400$ grid points.} 
\label{fig:WaveNumbersIllus}
\end{figure}
The discrete complex Fourier transform uses $n^2$ different wavenumbers $\vect k_j$ each
having a corresponding Fourier term $\exp(\imag \vect k_j^T \vect s)$. The real
Fourier transform, on the other hand, uses $n^2/2 +2$ different wavenumbers,
where four of them have only a cosine term and the others each have sine
and cosine terms. This follows from the fact that, for real data, certain coefficients of the complex transform
are the complex transpose of other coefficients. For technical details on
the real Fourier transform, we refer to \cite{DuMe84},
\cite{BoTaHa84}, \cite{RoWi05}, and \cite{Pa07}. Figure
\ref{fig:WaveNumbersIllus} illustrates an example of the spatial
wavenumbers, with $n^2=20\times 20=400$ grid points. The dots with a circle represent the
wavenumbers actually used in the real Fourier transform, and the red crosses mark the wavenumbers
having only a cosine term. Note that in \eqref{SolFT} we choose to order
the spatial wavenumbers such that the first four spatial wavenumbers
correspond to the cosine-only terms.
\begin{figure}[!ht]
\begin{center}
\includegraphics[width=\textwidth]{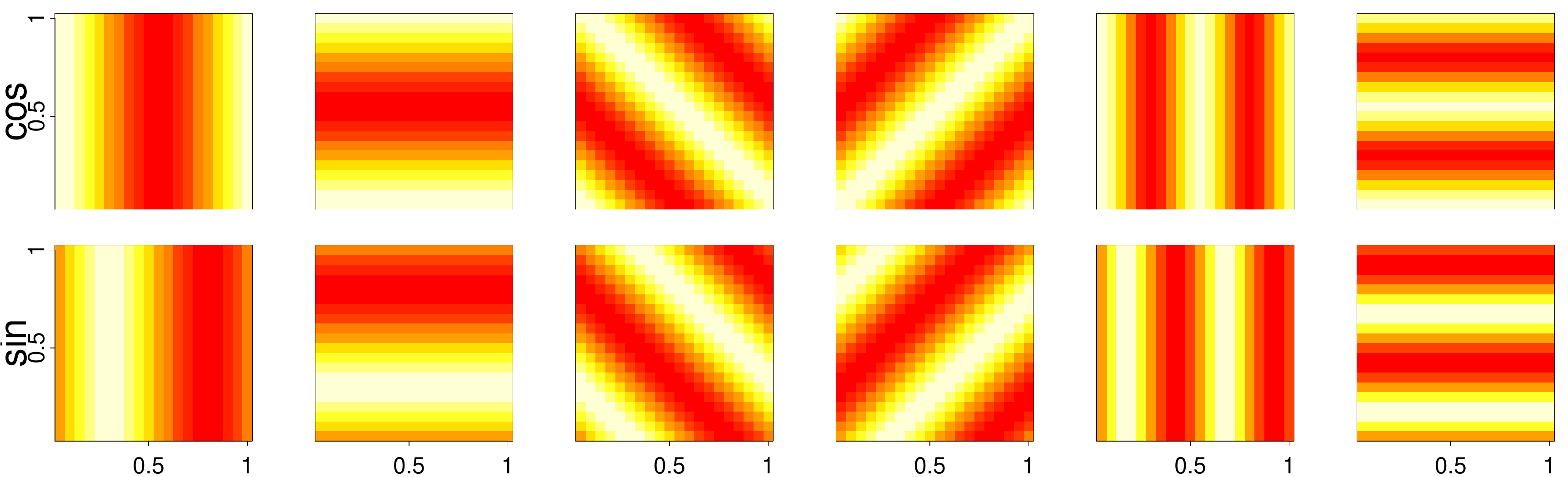} 
\end{center}
\caption{Illustration of two dimensional Fourier basis functions used in
  the discrete real Fourier transform with $n^2=400$. On the x- and y-axis
  are the coordinates of $\vect{s}$.} 
\label{fig:WaveIllus}
\end{figure}
To get an idea of what the basis functions $\cos{(\vect{k}_j^T\vect{s})}$ and
$\sin{(\vect{k}_j^T\vect{s})}$ look like, we plot in Figure
\ref{fig:WaveIllus} twelve low-frequency basis functions corresponding to the
six spatial frequencies closest to the origin $\vect 0$.
\begin{figure}[!ht]
  \centering
\includegraphics[width=0.5\textwidth]{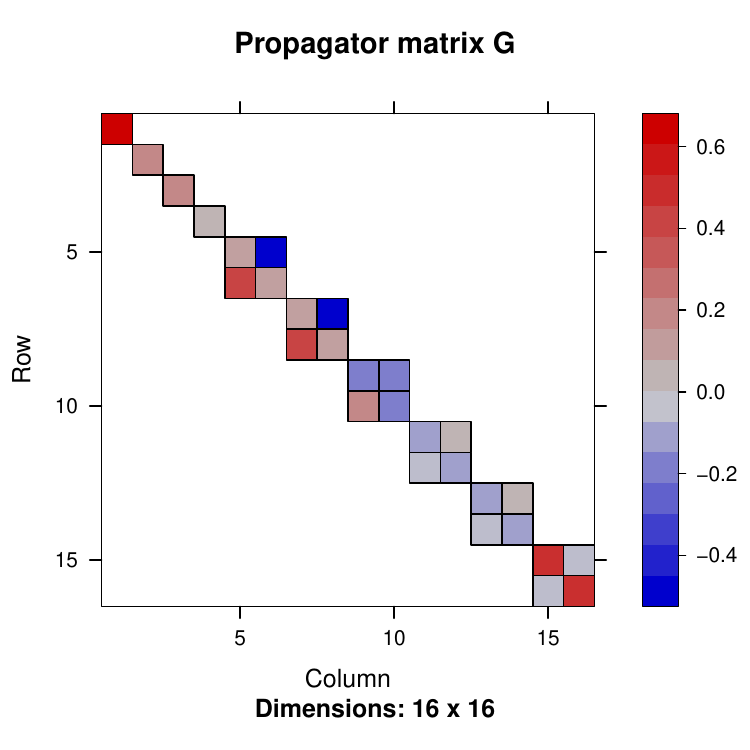}
\caption{Illustration of propagator matrix $\mat G$. 16 real Fourier functions
  are used ($n=4$).} 
\label{fig:Propagator}
\end{figure}
Further, in Figure \ref{fig:Propagator}, there is an example of a
propagator matrix $\mat G$ when $n=4$, i.e., when sixteen ($4^2$) spatial
basis functions are used. The upper left $4 \times 4$ diagonal matrix
corresponds to the cosine-only frequencies. The $2 \times 2$
blocks following correspond to wavenumbers with cosine - sine pairs.

Concerning notation in this paper, $K$ refers to the
number of Fourier terms, i.e., this is the dimension of the spectral
process $\vect \alpha(t)$ at each time $t$. Furthermore, $N$ denotes the number of
points at which the process $\vect \xi(t)$ is modeled, and $n$ is the
number of points on each axis of the quadratic grid used. Often, we have
$n^2=N=K$. However, if one uses a reduced dimensional Fourier basis, $K$ is smaller than $N$, see Section \ref{dimred}.



\subsection{Remarks on Finite Differences}\label{SpecvsFD}
Another approach to solve PDEs or SPDEs such as the one in
\eqref{SPDE} consists of using a discretization such as finite differences.  \cite{StEtAl10} use finite differences to solve an advection-diffusion PDE. Other examples are \cite{Wi03}, \cite{XuWi07},
\cite{DuGeSi}, \cite{MaEtAl08}, and \cite{ZhAu10}. The finite difference
approximation, however, has several disadvantages. First, each spatial
discretization effectively implies an interaction structure between
temporal and spatial correlation. In other words, as \cite{XuWiFo05} state, the discretization effectively suggests a
knowledge of the scale of interaction, lagged in time. Usually, this
space-time covariance interaction structure is not known, though. Furthermore, there
are numerical stability conditions that need to be fulfilled so that
the approximate solution is meaningful. Since these conditions depend on
the values of the unknown parameters, one can run into problems.

In addition, computational tractability is an issue. In fact, we have tried
to solve the SPDE in \eqref{SPDE} using finite differences as described in
the following. A finite difference approximation in \eqref{SPDE} leads to a vector
autoregressive model with a sparse propagator matrix being determined by the
discretization.  The innovation term $\epsilon$ can be approximated using a Gaussian
Markov random field with sparse precision matrix (see
\cite{LiLiRu10}). Even though the propagator and the precision matrices of
the innovations are sparse, we have run into a computational bottleneck
when using the Forward Filtering Backward Sampling (FFBS) algorithm
\citep{CaKo94, Fr94} for fitting the model. The basic
problem is that the Kalman gain is eventually a dense matrix. Alternative
sampling schemes like the information filter (see, e.g., \cite{AnMo79} and
\cite{ViFe09}) did not solve the problem either. However, future research
on this topic might come up with solutions.
 
\section{Computationally Efficient Statistical Inference}\label{inference}
The computational cost for one evaluation of the likelihood or one sample
from the full conditional in a spatio-temporal model with $T$ time points and $N$ spatial points equals
$O((NT)^3)$ when taking a naive approach. Using the Kalman filter or the Forward Filtering Backward
Sampling (FFBS) algorithm \citep{CaKo94, Fr94}, depending on what is
needed, this cost is reduced to $O(T N^3)$ which, generally, is still too
high for large data sets. In the following, we show how evaluation of the likelihood and sampling from the full conditional of the latent process can be
done efficiently in $O(TN \log N)$ operations. In the spectral space, the
costs of the algorithms grow linearly in the dimension $TN$, which means that
the total computational costs are
dominated by the costs of the fast Fourier transform (FFT) \citep{CoTu65}
which are $O(TN \log N)$. Furthermore, computational time
can be reduced by running the $T$ different FFTs in parallel. 

As is often done in a statistical model, we add a non-structured Gaussian term $\nu(t_{i+1},\vect
s) \sim N(0,\tau^2), iid,$ to \eqref{SolFTVec} to account for small scale variation and / or
measurement errors. In geostatistics, this term is called nugget
effect. Denoting the observations at time $t_i$ by $\vect{w}(t_i)$, we then have the following linear Gaussian state space model:
\begin{align}\label{GausObs}
\vect{w}(t_{i+1})&=\mat{\Phi}\vect{\alpha}(t_{i+1})+\vect
\nu(t_{i+1}), &\vect \nu(t_{i+1})\sim
N(0,\tau^2 \mat 1_N),\\
\vect{\alpha}(t_{i+1})&=\mat{G}\vect{\alpha}(t_i)+\widetilde{\vect{\epsilon}}(t_{i+1}),&\widetilde{\vect{\epsilon}}(t_{i+1})\sim
N(0,\widetilde{\mat{Q}}).\nonumber
\end{align}
Note that $\vect{\xi}(t_{i+1})=\mat{\Phi}\vect{\alpha}(t_{i+1})$. As
mentioned before, irregular spatial data can be modeled by adopting a data
augmentation approach (see \cite{SiKuSt11}) or by using an incidence matrix (see Section \ref{dimred}). For the sake of simplicity, a zero mean was assumed. Extending the model by
including covariates in a regression term is straightforward. Furthermore, we assume
normality. The model can be easily generalized to allow for data not following a
Gaussian distribution. For instance, this can be done by
including it in a Bayesian hierarchical model (BHM) \citep{WiBeCr98} and
specifying a non-Gaussian distribution for $\vect{w}|\vect{\xi}$. The
posterior can then no longer be evaluated exactly. But
approximate posterior probabilities can still be computed using, for instance, simulation
based methods such as Markov chain Monte Carlo (MCMC) (see, e.g.,
\cite{GilWRS96} or \cite{RobCC04}). An additional
advantage of BHMs is that these models can be extended, for
instance, to account for temporal non-stationarity by letting one or several
parameters vary over time.

\subsection{Kalman Filtering and Backward Sampling in the Spectral Space}\label{kfbsss}
When following both a frequentist or a Bayesian paradigm, it is crucial that one is able to evaluate the
likelihood of the hyper-parameters given $\vect w$ with a reasonable computational effort. In addition, when doing Bayesian
inference, one needs to be able to simulate efficiently from the full
conditional of the latent process $[\vect \xi|\cdot]$, or, equivalently,
the Fourier coefficients $[\vect \alpha|\cdot]$. Below, we show
how both these tasks can be done in the spectral space in linear time,
i.e., using $O(TN)$ operations. For transforming between the physical and
spectral space, one can use the FFT which requires $O(TN \log N)$
operations. We start with the spectral version of the Kalman filter. Its
output is used for both evaluating the log-likelihood and for simulating
from the full conditional of the coefficients $\vect \alpha$. 

\begin{algorithm}
\caption{Spectral Kalman filter}\label{skf}
  \textbf{Input:} $T$,$\widetilde{\vect w}$, $\mat G$, $\tau^2$,
  $\widetilde{\mat{Q}}$, $\mat F$ \\
\textbf{Output:} forecast and filter means $\vect m_{t_i|t_{i-1}}$, $\vect
m_{t_i|t_{i}}$ and covariance matrices $\mat
R_{t_i|t_{i}}$, $\mat R_{t_i|t_{i-1}}$, $i=1,\dots,T$

\begin{algorithmic}
\STATE $ \vect m_{t_0|t_0}= \vect 0$
\STATE $\mat R_{t_0|t_0}=\widetilde{\mat{Q}}$
\FOR{$i=1$ to $T$}
\STATE  $\vect m_{t_i|t_{i-1}}=\mat G \vect m_{t_{i-1}|t_{i-1}}$
\STATE $\mat R_{t_i|t_{i-1}}=\widetilde{\mat{Q}} + \mat R_{t_{i-1}|t_{i-1}}
\mat F$
\STATE $\mat R_{t_i|t_{i}}=\left(\tau^{-2}\mat 1_N+\mat R_{t_i|t_{i-1}}^{-1}\right)^{-1}$
\STATE  $\vect m_{t_i|t_i}=\vect m_{t_i|t_{i-1}}+\tau^{-2}\mat
R_{t_i|t_i}\left( \widetilde{\vect w}(t_i)- \vect m_{t_i|t_{i-1}}\right)$
\ENDFOR
\end{algorithmic}
\end{algorithm}
Algorithm \ref{skf} shows the Kalman filter in the spectral space. For the sake of simplicity, we assume that the initial distribution equals the
innovation distribution. The spectral Kalman filter has as input the Fourier transform of
$\widetilde{\vect w}=(\widetilde{\vect w}(t_1)^T,\dots,\widetilde{\vect
  w}(t_T)^T)^T$ of $\vect w$, the diagonal matrix $\mat F$ given by
\begin{equation}\label{Hdef}
\begin{split}
[\mat F]_{1:4,1:4}&=\textnormal{diag}\left(\expo{ -2\Delta
    (\vect{k}_j^T\mat{\Sigma}\vect{k}_j+ \zeta)}\right),\\
[\mat F]_{5:N,5:N}&=\textnormal{diag}\left(\expo{-2\Delta(\vect{k}_j^T\mat{\Sigma}\vect{k}_j +
    \zeta)}\mat 1_2\right),
\end{split}
\end{equation}
  and other parameters that characterize the SPDE
model. It returns forecast and filter means $\vect
m_{t_i|t_{i-1}}$ and $\vect m_{t_i|t_{i}}$ and covariance matrices
$\mat R_{t_i|t_{i}}$ and $\mat R_{t_i|t_{i-1}}$, $i=1,\dots,T$,
respectively. I.e., $\vect m_{t_i|t_{i}}$ and $\mat R_{t_i|t_{i}}$ are the
mean and the covariance matrix of $\vect \alpha (t_i)$ given data up to
time $t_i$ $\{\vect
w(t_j)|j=1,\dots,i\}$. Analogously,  $m_{t_i|t_{i-1}}$ and $\mat
R_{t_i|t_{i-1}}$ are the forecast mean and covariance matrix given data up
to time $t_{i-1}$. We follow the notation of \cite{Ku01}. 

Since the matrices $\widetilde{\mat{Q}}$ and $\mat F$ are diagonal, the
covariance matrices $\mat R_{t_i|t_{i}}$ and $\mat R_{t_i|t_{i-1}}$ are
also diagonal. Note that the matrix notation in Algorithm \ref{skf} is used
solely for illustrational purpose. In practice, matrix vector products
($\mat G \vect m_{t_{i-1}|t_{i-1}}$), matrix multiplications ($\mat
R_{t_{i-1}|t_{i-1}}\mat F$), and matrix inversions $(\tau^{-2}+\mat
R_{t_i|t_{i-1}})^{-1}$ are not calculated with general purpose algorithms
but elementwise since all matrices are diagonal or $2 \times 2$ block diagonal. It follows that the
computational cost for this algorithm is $O(TN)$.

The derivation of this algorithm follows from the classical Kalman filter
(see, e.g., \cite{Ku01}) using $\mat \Phi' \mat \Phi =\mat 1_N$, $\mat G\mat
R_{t_{i-1}|t_{i-1}} \mat G^T=\mat R_{t_{i-1}|t_{i-1}} \mat G \mat G^T$, and the fact
that $\mat G \mat G^T=\mat F$. The first equation holds true due to the
orthonormality of the discrete Fourier transform. The second equation
follows from the fact that $\mat G$ is $2 \times 2$ block diagonal and that $\mat
R_{t_{i-1}|t_{i-1}}$ is diagonal with the diagonal entries being equal for each cosine - sine pair. The last equation holds true as shown in the
following. Being obvious for the first four frequencies, we consider the
$2 \times 2$ diagonal blocks of cosine - sine pairs:
\begin{equation*}
\begin{split}
&[\mat G]_{(2l-5):(2l-4),(2l-5):(2l-4)}[\mat
G]_{(2l-5):(2l-4),(2l-5):(2l-4)}^T\\ &=\expo{-2\Delta(\vect{k}_j^T\mat{\Sigma}\vect{k}_j +
    \zeta)}\left(\cos(\Delta\vect{\mu}^T\vect{k}_j)\mat 1_2-\sin(\Delta\vect{\mu}^T\vect{k}_j)\mat{J}_2\right)\left(\cos(\Delta\vect{\mu}^T\vect{k}_j)\mat 1_2-\sin(\Delta\vect{\mu}^T\vect{k}_j)\mat{J}_2\right)^T\\
&=\expo{-2\Delta(\vect{k}_j^T\mat{\Sigma}\vect{k}_j + \zeta)}\left(\cos(\Delta\vect{\mu}^T\vect{k}_j)^2+\sin(\Delta\vect{\mu}^T\vect{k}_j)^2 \right)\mat{1}_2,
\end{split}
\end{equation*}
$l=5,\dots, N/2+2$, which equals \eqref{Hdef}. In the last equation we have used
$$\mat J_2^T=-\mat J_2 ~~\textnormal{and}~~ \mat J_2^2=-\mat 1_2.$$

Based on the Kalman filter, the log-likelihood is calculated as (see, e.g., \cite{ShSt00})
\begin{equation}\label{ll}
\begin{split}
\ell =& \sum_{i=1}^T{\log\left|\mat R_{t_i|t_{i-1}}+\tau^2\mat 1_N\right|+\left(\widetilde{\vect
      w}(t_i)-\vect m_{t_i|t_{i-1}}\right)^T\left(\mat
    R_{t_i|t_{i-1}}+\tau^2\mat 1_N\right)^{-1}\left(\widetilde{\vect
      w}(t_i)-\vect m_{t_i|t_{i-1}}\right)}\\ &+\frac{TN}{2}\log(2\pi).
\end{split}
\end{equation}
Since the forecast covariance matrices $\mat R_{t_i|t_{i-1}}$ are diagonal,
calculation of their determinants and their inverses is trivial, and
computational cost is again $O(TN)$.

\begin{algorithm}
\caption{Spectral backward sampling}\label{sbs}
        \textbf{Input:} $T$, $\mat G$,  $\widetilde{\mat{Q}}$, $\mat F$, $\vect m_{t_i|t_{i-1}}$, $\vect m_{t_i|t_{i}}$, $\mat R_{t_i|t_{i}}$, $\mat R_{t_i|t_{i-1}}$, $i=1,\dots,T$\\
\textbf{Output:} a sample $\vect{\alpha}^*(t_1),\dots,\vect{\alpha}^*(t_T)$
from $[\vect \alpha|\cdot]$

\begin{algorithmic}
\STATE  $\vect{\alpha}^*(t_T)=\vect m_{t_T|t_{T}}+\left(\mat
  R_{t_T|t_{T}}\right)^{1/2} \vect n_T,~~ \vect n_T \sim N(\vect 0,\mat 1_N)$
\FOR{$i=T-1$ to $1$}
\STATE  $\overline{\vect m}_{t_i}=\vect m_{t_i|t_{i}}+\mat
R_{t_i|t_{i}}\mat R_{t_i|t_{i-1}} ^{-1}\mat
G^T\left(\vect{\alpha}^*(t_{i+1})-\vect m_{t_i|t_{i-1}}\right)$
\STATE  $\overline{\mat R}_{t_i}=\left(\widetilde{\mat{Q}} \mat F + \mat
  R_{t_{i-1}|t_{i-1}}^{-1}\right)^{-1}$
\STATE  $\vect{\alpha}^*(t_i)=\overline{\vect m}_{t_i} +\left(\overline{\mat R}_{t_i}\right)^{1/2} \vect n_i,~~ \vect n_i \sim N(\vect 0,\mat 1_N)$
\ENDFOR
\end{algorithmic}
\end{algorithm}
In a Bayesian context, the main difficulty consists in simulating from the full
conditional of the latent coefficients $[\vect \alpha|\cdot]$. After running the Kalman filter, this
can be done with a backward sampling step. Together, these two algorithms
are know as Forward Filtering Backward Sampling (FFBS) \citep{CaKo94, Fr94}. Again, backward sampling is computationally very efficient in the spectral space with cost being $O(TN)$. Algorithm \ref{sbs} shows the backward sampling
algorithm in the spectral space. The matrices $\overline{\mat R}_{t_i}$ are diagonal which makes
their Cholesky decomposition trivial.

\subsection{Dimension Reduction and Missing or Non-Gridded Data}\label{dimred}
If desired, the total computational cost can be additionally alleviated by
using a reduced dimensional Fourier basis with $K<<N$, $N$ being the number
of grid points. This means that one includes only
certain frequencies, typically low ones. When the Fourier
transform has been made, the spectral filtering and
sampling algorithms then require $O(KT)$ operations. For using the FFT, the
frequencies being excluded are just set to zero. Performing the FFT still
requires $O(TN\log N)$ operations, though.

When the observed data does not lie on a grid or has missing data, there
are two alternative approaches. First, one can use a data augmentation approach \citep{SmRo93}
for the missing data. See Section \ref{Fitting} and, for more details,
\cite{SiKuSt11}. For irregularly spaced data, one can assign the
data to a regular grid and treat the cells with no observations as missing data. FFT can then be applied to the augmented data, and the
algorithms presented above can be used. Alternatively, as is the case in our application, one can
include an incidence matrix $\mat H$ that relates the process on the grid to the
observation locations. Instead of \eqref{GausObs}, the model is then
\begin{equation}\label{Incidence}
\vect{w}(t_{i+1})=\mat H \mat{\Phi}\vect{\alpha}(t_{i+1})+\vect
\nu(t_{i+1}), ~~\vect \nu(t_{i+1})\sim
N(0,\tau^2   \mat{1}_N).
\end{equation} 
However, in the Kalman filter, the term  $(\mat H \mat{\Phi})^T\mat H
\mat{\Phi}$, used for calculating the filter covariance matrix $\mat R_{t_i|t_{i}}$, is not a
diagonal matrix anymore. From this follows that the
Kalman filter does not diagonalize in the spectral space if one uses an incidence matrix $\mat H$. Consequently, one has to
use the traditional FFBS for which computational cost is $O(K^3T)$. This
means that dimension reduction is required to make this approach
computationally feasible.

\subsection{An MCMC Algorithm for Bayesian Inference}\label{MCMC}
Based on the algorithms presented above, there are different possible ways for doing
statistical inference. For instance, if one adopts a frequentist paradigm, one can numerically
maximize the log-likelihood in \eqref{ll}. In the following, we briefly
present how Bayesian inference can be done using a Monte Carlo Markov Chain
(MCMC) algorithm \citep[see][]{GilWRS96, RobCC04, brooks2011handbook}. This algorithm is implemented in the R
package \func{spate} \citep{SiKuSt12b} and used in the application
in Section \ref{Postproc}. 

To complete the specification of a Bayesian model, prior
distributions for the  parameters $\vect
\theta=(\rho_0,\sigma^2,\zeta,\rho_1,\gamma,\alpha,\mu_x,\mu_y,\tau^2)^T$
have to be chosen. In general, this choice can depend on the specific
application. We present choices for priors that are weakly uninformative. Based on \cite{Ge06}, we suggest to use improper
priors for the $\sigma^2$ (marginal variance of the innovation) and
$\tau^2$ (nugget effect variance) that are uniform on the standard deviation scale $\sigma$ and
$\tau$, respectively. Further, the drift parameters $\mu_x$ and $\mu_y$ have uniform priors on
$[-0.5,0.5]$, $\psi$ (direction of anisotropy) has a uniform prior on
$[0,\pi/2]$, and $\gamma$ (degree of anisotropy) has
a uniform prior on the log scale of the interval $[0.1,10]$. $\gamma$ is restricted to $[0.1,10]$ since stronger anisotropy does
not seem reasonable. The range parameters of the innovations and the
diffusion matrix $\rho_0$ and $\rho_1$, respectively, as well as
the damping parameter $\zeta$ are assigned improper, locally uniform priors
on $\mathbb{R}_+$.

Our goal is then to simulate from the joint posterior of the unobservables $[\vect{\theta},\vect{\alpha}|\vect{w}]$, where
$\vect{w}$ denotes the set of all observations. Missing data can be
accommodated for by using a data augmentation approach which results in an
additional Gibbs step, see Section \ref{Fitting}. Since the latent process $\vect \xi$
is the Fourier transform of the coefficients $\vect \alpha$, $\vect
\xi(t_i)=\mat \Phi \vect \alpha(t_i)$, sampling from posterior of $\vect \alpha$ is, from a methodological point of view, equivalent to sampling from the one of $\vect \xi$. In the following, we use the notation $[w|\cdot]$ and $P[w|\cdot]$  to
denote conditional distributions and densities, respectively. 

A straightforward approach would be to sample iteratively from the full
conditionals of $\vect{\theta}$ and $\vect{\alpha}$. One could also further divide the
latent process $\vect{\alpha}$ in blocks by iteratively sampling $\vect{\alpha}(t_i)$ at each
time point. However $\vect{\theta}$ and $\vect{\alpha}$ can be
strongly dependent which results in slow mixing. This problem is similar to the one observed when doing
inference for diffusion models, see, e.g., \cite{RoSt01} and
\cite{GoWi08}. It is therefore recommendable to sample jointly from $[
\vect{\theta},\vect{\alpha}|\vect w]$ in a Metropolis-Hastings step.

Joint sampling from $\vect{\theta}$ and $\vect{\alpha}$ is done as
follows. First, a proposal
$(\vect{\theta}^*,\vect{\alpha}^*)$ is obtained by sampling
$\vect{\theta}^*$ from a Gaussian distribution with the
mean equaling the last value and an adaptively estimated proposal
covariance matrix. To be more specific, $\rho_0,\sigma^2,\zeta,\rho_1,\gamma$, and $\tau^2$ are sampled on a log scale to ensure that they remain positive. Then, a sample $\vect{\alpha}^*$ from
$[\vect{\alpha}|\vect{\theta}^*,\vect w]$ is obtained using the forward filtering backward sampling
(FFBS) algorithm \citep{CaKo94, Fr94}. It can be shown that the acceptance ratio for the joint
proposal is
  \begin{equation}\label{AccRat}
 \min\left(1,\frac{P[\vect{\theta}^*|\vect{w}]P[\vect{\theta}^*] \rho_0^* \sigma^{2*} \zeta^* \rho_1^* \gamma^* \tau^{2*} }{P[\vect{\theta}^{(i)}|\vect{w}]P[\vect{\theta}^{(i)}]\rho_0^{(i)} \sigma^{2{(i)}} \zeta^{(i)} \rho_1^{(i)} \gamma^{(i)} \tau^{2{(i)}}}\right),
  \end{equation}
where $P[\vect{\theta}|\vect{w}]$ denotes the likelihood of
$\vect{\theta}$ given $\vect{w}$, $P[\vect{\theta}]$ the prior, and where $\vect{\theta}^*$ and $\vect{\theta}^{(i)}$
denote the proposal and the last values, respectively. The factor $\rho_0\sigma^{2} \zeta \rho_1 \gamma \tau^{2}$ is included since these parameters are sampled on a log scale. We see that the above
acceptance ratio does not depend on  the latent process
$\vect \xi = \mat \Phi \vect{\alpha}$. Thus, the parameters $\vect{\theta}$ are allowed to move
faster in their parameter space. The value of the likelihood
$P[\vect{\theta}|\vect{w}]$ is obtained as a side product
of the Kalman filter in the FFBS. 

For this random walk Metropolis step, we suggest to use an adaptive algorithm
\citep{RoRo09} meaning that the proposal covariance matrices for $\vect{\theta}$ are successively estimated such that an optimal scaling is obtained with an acceptance rate between $0.2$ and $0.3$. See \cite{RoRo01} for more information on optimal scaling for Metropolis-Hastings algorithms. 

In
addition, if the model includes a regression term (see
the application in Section \ref{Postproc}), the fixed effects can also be
strongly dependent with the random effects $\vect \xi$. This means that it is advisable that the coefficients $\vect
b \in \mathbb{R}^p$ of the potential covariates $\vect x(t, \vect s) \in
\mathbb{R}^p$ are also sampled together with $\vect{\theta}$ and
$\vect{\alpha}$. This can be done by slightly modifying the above
algorithm. First, the regression coefficients $\vect{b}^*$ are proposed
jointly with $\vect{\theta}^*$ in a random walk Metropolis step. Then
$\vect{\alpha}^*$ is sampled from
$[\vect{\alpha}|\vect{\theta}^*,\vect{b}^*,\vect w]$ analogously using the
FFBS. Finally, in the acceptance ration in \eqref{AccRat}, $P[\vect{\theta}|\vect{w}]$ now just has to be
replaced by $P[\vect{\theta},\vect b|\vect{w}]$ which is also a side
product of the Kalman filter.

\section{\titleq{Postprocessing Precipitation Forecasts}}\label{Postproc}
Numerical weather prediction (NWP) models are capable of
producing predictive fields at spatially and temporally high
frequencies. Statistical postprocessing, which is the main objective of this application, serves two purposes. First,
probabilistic predictions are obtained in cases where only deterministic
ones are available. Further, even if ``probabilistic'' forecasts in form of
ensembles \citep{Pa02, GnRa05} are available, they are typically
not calibrated, i.e., they are often underdispersed \citep{HaCo97}. The
goal of postprocessing is then to obtain calibrated and sharp predictive
distributions (see \cite{GnBaRa07} for a definition of calibration and
sharpness). In the case of precipitation, the need for postprocessing is
particularly strong, since, despite their importance, precipitation
forecasts are still not as accurate as forecasts for other meteorological
quantities \citep{ApEtAl02, StYu07}. 

Several approaches for postprocessing precipitation forecasts have been
proposed, including linear regression \citep{An00}, logistic regression
\citep{HaWhWe04}, quantile regression \citep{Br04, FrHe07},
hierarchical models based on a prior climatic distribution \citep{KrMa06},
neural networks \citep{RaVeFe05}, and binning techniques \citep{YuSt06}. \cite{SlRaGn07} propose a
two-stage model to postprocess precipitation forecasts. \cite{BeRaGn08}
extended the model of \cite{SlRaGn07} by accounting for spatial
correlation. \cite{KlRaGn11} present a similar model that includes
ensemble predictions and accounts for spatial correlation.

Except for the last two references, spatial correlation is typically not modeled in postprocessing precipitation forecasts, and none of the
aforementioned models explicitly accounts for spatio-temporal dependencies. However,
for temporally and spatially highly resolved data, it is necessary to
account for correlation in space and time. First, spatio-temporal correlation is
important, for instance, for predicting precipitation accumulation over
space and time
with accurate estimates of precision. Further, it is likely that errors of
NWP models exhibit structured behaviour over space and time, including interactions between
space and time. The SPDE approach allows for such interactions, as do other approaches
which use scientifically-based physical models \citep{WiHo10}.

\subsection{Data}

 \begin{figure}[!ht]
\begin{center}
\includegraphics[width=0.65\textwidth]{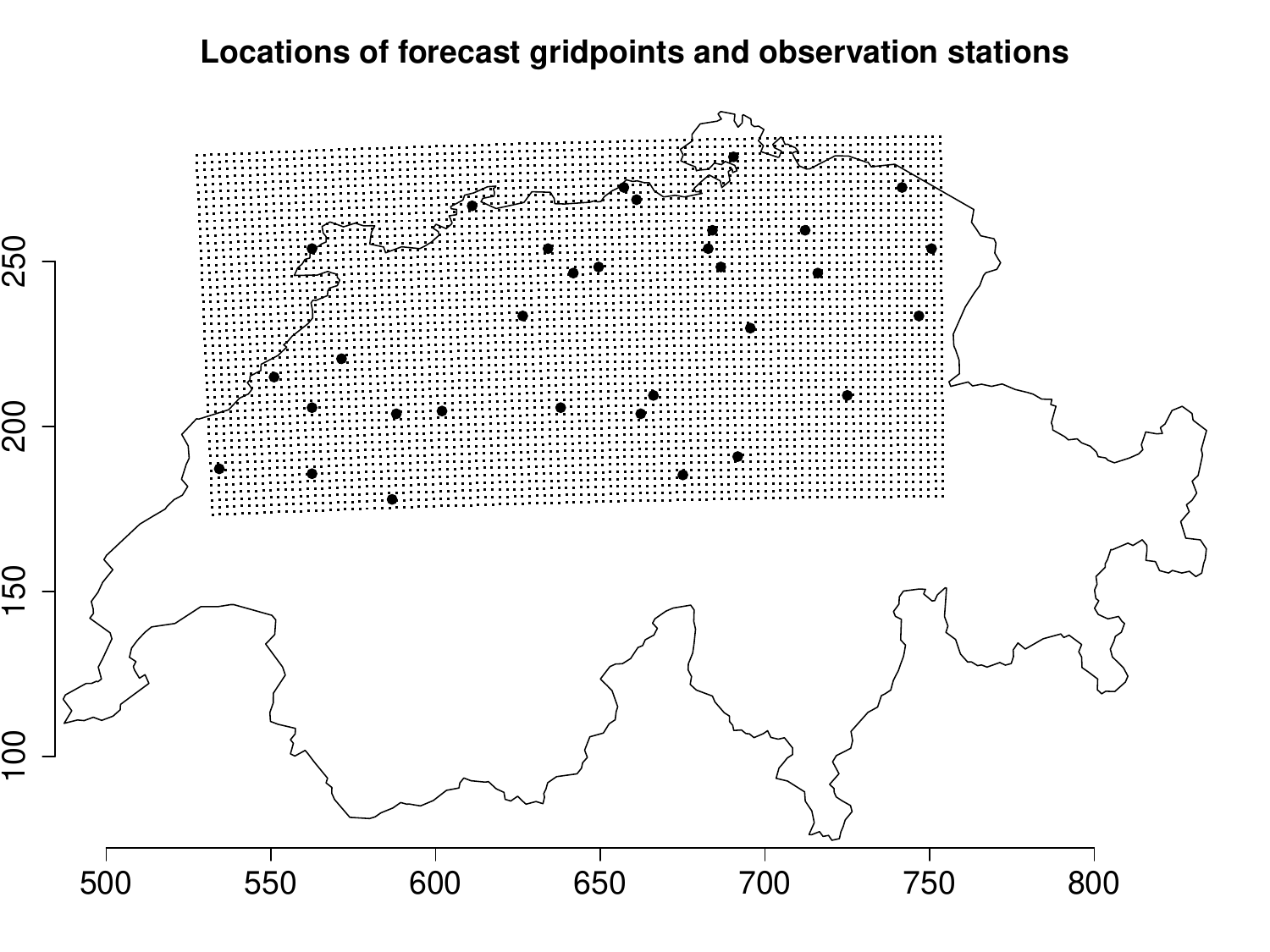} 
\end{center}
\caption{Locations of grid points at which predictions are obtained ($50
  \times 100$ grid of small dots) and
  observations stations (bold dots). Both axis are in km using the Swiss
  coordinate system (CH1903).} 
\label{fig:Stations}
\end{figure}
The goal is to postprocess precipitation forecasts from an NWP model called COSMO-2, a
high-resolution model with a grid spacing of 2.2 km that is run by
MeteoSwiss as part of COonsortium for Small-scale MOdelling (COSMO)
\citep[see, e.g.,][]{StEtAl03}. The NWP model produces deterministic
forecasts once a day
starting at 0:00UTC. Predictions are made for eight consecutive time
periods corresponding to 24 h ahead. In the following, let
$y_{F}(t,\vect{s})$ denote the forecast of the rainfall sum from time $t-1$
to $t$ at site $s$ made at 0:00UTC of the same day. We
consider a rectangular region in northern Switzerland shown in Figure
\ref{fig:Stations}. The grid at which predictions are made is of size $50
\times 100$. Precipitation is
observed at 32 stations over northern Switzerland. Figure
\ref{fig:Stations} also shows the locations of the observation
stations. In the postprocessing model, the NWP forecasts are used as
covariates in a regression term, see \eqref{ppmodel}. We use data for three-hourly rainfall amounts from the beginning of
December 2008 till the end of March 2009. To illustrate the observed data, in Figure \ref{fig:RainVsTime},
observed precipitation at one station and the equally weighted areal average precipitation
are plotted versus time. We will use the first three months containing
720 time points for fitting, and the last month is left aside for
evaluation.
\begin{figure}[!ht]
\begin{center}
\includegraphics[width=0.65\textwidth]{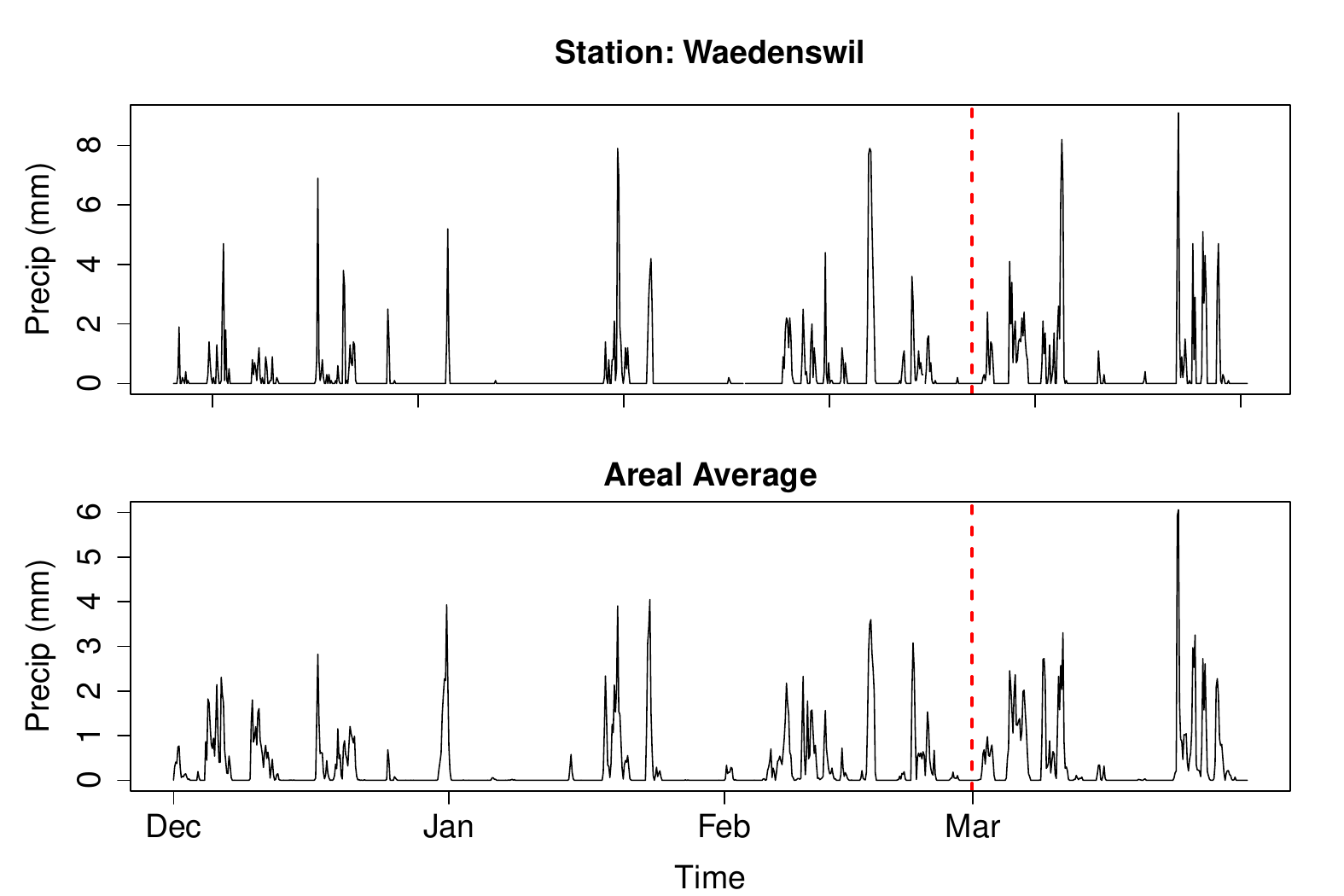} 
\end{center}
\caption{Precipitation (mm) versus time,
for one station and averaged over all stations.} 
\label{fig:RainVsTime}
\end{figure}

The NWP model forecasts are deterministic and ensembles are not available in our
case. However, the extension to use an ensemble instead of just one member can
be easily done. One can include all the ensemble members in the
regression part of the model. Or, in the case of exchangeable members,
one can use the location and the spread of the ensemble.

\subsection{\titleq{Precipitation Model for Postprocessing}}\label{Application}
The model presented in the following is a Bayesian hierarchical model
(BHM). It uses the SPDE based spatio-temporal Gaussian process $\xi(t,
\vect s)$ presented in Section \ref{SpecSpace} at the process level. At the data stage, a mixture model adapted to the nature of precipitation is
used. A characteristic feature of precipitation is that its distribution consists
of a discrete component, indicating occurrence of precipitation,
and a continuous one, determining the amount (see Figure \ref{fig:RainVsTime}). As a consequence, there are two
basic statistical modeling approaches. The continuous and the discrete part are either
modelled separately \citep{CoSt82, Wi99} or together
\citep{Be87, Wi90, BaPl92, Hu95, SaGu04}. See, e.g.,
\cite{SiKuSt11} for a more extensive overview of precipitation models and
for further details on the data model used below. Originally, the approach presented in the following
goes back to \cite{To58} who analyzed household expenditure
on durable goods. For modeling precipitation, \cite{St73} took up this idea
and modified it by including a power transformation for the non-zero part so that the model
can account for skewness. \cite{SaGu99} develop
Bayesian methods for the spatio-temporal analysis of rainfall using this skewed Tobit model, but in contrast to our application
they do not explicitly account for temporal correlation and they use a much smaller spatial grid.


We denote the cumulative rainfall from time $t-1$ to
$t$ at site $\vect s \in \mathbb{R}^{2}$ by $y(t,\vect{s})$ and assume that it depends
on a latent Gaussian variable $w(t,\vect{s})$ through
\begin{equation}\label{rainrel}
\begin{split}
  y(t,\vect{s})&=0 ,\hspace{1.5cm} \text{if}~w(t,\vect{s}) \leq 0,\\
  &=w(t,\vect{s})^{\lambda},\hspace{0.4cm} \text{if}~w(t,\vect{s})>0,\\
\end{split}
\end{equation}
where $\lambda>0$. A power transformation is needed since precipitation
amounts are skewed and do not follow a truncated normal distribution. The
latent Gaussian process $w(t,\vect{s})$ is interpreted as a precipitation potential. 


The mean of the Gaussian process $w(t,\vect{s})$ is assumed to depend linearly on
spatio-temporal covariates
$\vect{x}(t,\vect{s})\in \mathbb{R}^{k}$. As shown below, this mean term
basically consists of the NWP forecasts. Variation that is not explained
by the linear term is modeled using the Gaussian process $\xi(t,\vect{s})$
and the unstructured term $\nu(t,\vect{s})$ for microscale variability and
measurement errors. The spatio-temporal process $\xi(t,\vect{s})$ has two
functions. First, it captures systematic errors of the NWP in space and
time and can extrapolate
them over time. Second, it accounts for structured variability so that the
postprocessed forecast is probabilistic and its distribution sharp and
calibrated. 

To be more specific concerning the covariates, similarly to what appears in \cite{BeRaGn08}, we include a transformed variable $y_{F}(t,\vect{s})^{1/\tilde{\lambda}}$
and an indicator variable $\mathbbm{1}_{\{y_{F}(t,\vect{s})=0\}}$ which
equals $1$ if $y_{F}(t,\vect{s})=0$ and $0$ otherwise. $\tilde{\lambda}$ is
determined by fitting the transformed Tobit model as in \eqref{rainrel} to
the marginal distribution of the rain data ignoring any spatio-temporal
correlation. In doing so, we obtain $\tilde{\lambda}\approx 1.4$. $y_{F}(t,\vect{s})^{1/\tilde{\lambda}}$ is centered around zero by subtracting its
overall mean $\overline{y}^{1/\tilde{\lambda}}_{F}$ in order to reduce
posterior correlations. Thus, $w(t,\vect{s})$  equals
\begin{equation}\label{ppmodel}
w(t,\vect{s})=b_1\left(y_{F}(t,\vect{s})^{1/\tilde{\lambda}}-\overline{y}_F^{1/\tilde{\lambda}}\right)+b_2\mathbbm{1}_{\{y_{F}(t,\vect{s})=0\}}+\xi(t,\vect{s})+\nu(t,\vect{s}).
\end{equation}
An intercept is not included since the first Fourier term is constant in
space. In our case, including an intercept term results in weak identifiability
which slows down the convergence of the MCMC algorithm used for fitting. Note that in
situations where the mean is large it is advisable to include an intercept,
since the coefficient of the first Fourier term is constrained by the joint
prior on $\vect \alpha$. Further, unidendifiability is unlikely to be a
problem in these cases.

Concerning the spatio-temporal process $\xi(t,\vect{s})$, we apply
padding. This means that we embed the $50\times 100$ grid in a rectangular $200
\times 200$ grid. A brief prior investigation showed that the range parameters are relatively large in comparison
to the spatial domain, and padding is therefore used in order to avoid
spurious correlations due to periodicity. The NWP forecasts are
not available on the extended $200 \times 200$ domain, which means that, in principle, the process
$w(t,\vect{s})$ can only be modeled on the $50\times 100$ grid where the
covariates are available. To cope with this we use an incidence
matrix $\mat H$ as in \eqref{Incidence} to relate the process at the $200 \times 200$ grid
to the observation stations. As argued in Section \ref{dimred}, this then requires that we use
a reduced dimensional Fourier expansion. I.e., instead of using $N=200^2$
basis functions, we only use $K<<N$ low-frequency Fourier terms. Since the
observation stations are relatively scarce, one might argue that there is no information on spatial high frequencies of the NWP
error, and that the high frequencies can be left out. In fact, this
hypothesis gets confirmed by our analysis, see Figure \ref{fig:CRPS}.

Concerning prior distributions, for $\vect
\theta=(\rho_0,\sigma^2,\zeta,\rho_1,\gamma,\psi,\mu_x,\mu_y,\tau^2)^T$, we
use the priors presented in Section \ref{MCMC}. The parameters $\vect{b}$
and $\lambda$, which are not included in $\vect \theta$, have improper, locally uniform priors on $\mathbb{R}$
and $\mathbb{R}_+$, respectively.  In summary, 
$$P[\vect b, \lambda, \vect \theta] \propto \frac{1}{\sqrt{\sigma^2}\sqrt{\tau^2}\gamma} \mathbbm{1}_{\{-0.5\leq \mu_x,\mu_y\leq 0.5\}}
\mathbbm{1}_{\{0\leq \psi \leq
  \pi/2\}}\mathbbm{1}_{\{\lambda,\rho_0,\rho_1,\zeta,\sigma^2,\tau^2\geq0\}}\mathbbm{1}_{\{0.1\leq
  \gamma \leq 10\}}.$$
In addition, concerning $\vect \alpha(0)$, we choose to use the innovation
distribution specified in \eqref{WhittleSpec} as initial distribution. 

\subsection{\titleq{Fitting}}\label{Fitting}
Monte Carlo Markov Chain (MCMC) is used to sample from the posterior
distribution $[\vect{b},\lambda,\vect{\theta},\vect{\alpha},\vect{w}|\vect{y}]$, where
$\vect{y}$ denotes the set of all observations. We use what \cite{NeRo06} call a
Metropolis within-Gibbs algorithm which alternates between blocked Gibbs
\citep{GeSm90} and Metropolis \citep{MeetAl53, Ha70} sampling steps. 

We use the Metropolis-Hastings algorithm presented in Section
\ref{MCMC} with the coefficients $\vect b$ being sampled jointly with
$\vect \theta$ and $\vect \alpha$. Due to the non-Gaussian data model, additional Metropolis and Gibbs steps are
required for $\lambda$ and for those points of $\vect{w}$ where the
observed rainfall amount is zero and where observations are missing. We
refer to \cite{SiKuSt11} for more details on the type of data
augmentation approach that is used for doing this. We denote by
$\vect{w}^{[0]}$ the values of $\vect{w}$ at those points where the
observed rainfall is zero, $y(t,\vect s)=0$. Analogously, we define $\vect{w}^{[m]}$ and
$\vect{w}^{[+]}$ for the missing values and the values
where a positive
rainfall amount is observed, $y(t,\vect s)>0$, respectively. The full conditionals of the
censored $\vect{w}^{[0]}$ and missing points $\vect{w}^{[m]}$ are truncated
and regular one-dimensional Gaussian distributions, respectively. Sampling
from them is done in Gibbs steps. The
transformation parameter $\lambda$ is sampled using a random walk
Metropolis step. If a new value is accepted, $\vect{w}^{[+]}$ needs
to be updated using the deterministic relation $w(t,\vect s)=y(t,\vect
s)^{1/\lambda}$ due to \eqref{rainrel}. 
From these Gibbs and
Metropolis steps, we obtain $\vect{w}$ consisting of simulated and
transformed observed data. In the second part of the algorithm, we sample
$\vect b, \vect{\theta}$, and $\vect{\alpha}$
jointly from $[\vect b, \vect{\theta},\vect{\alpha}|\vect{w}]$ using the
algorithm presented in Section
\ref{MCMC}, where $\vect{w}$ acts as if it was the observed
data. After a burn-in of $5,000$ iterations,
we use $100,000$ samples from the Markov chain to characterize the
posterior distribution. Convergence is monitored by inspecting trace
plots.

\subsection{Model Selection and Results}
We use a reduced dimensional approach. The
number of Fourier functions is determined based on predictive
performance for the $240$ time points that were set aside. We start with models including only low spatial frequencies and add
successively higher frequencies. In doing so, we only consider models that
have the same resolution in each direction, i.e., we do not
consider models that have higher frequency spatial basis functions in the
east-west direction than in the north-south one.

In order to assess the performance of the predictions and to choose the number of basis functions to include, we use the continuous
ranked probability score (CRPS) \citep{MaWi76}. The
CRPS is a strictly proper scoring rule \citep{GnRa07} that assigns a
numerical value to probabilistic forecasts and assesses calibration and
sharpness simultaneously \citep{GnBaRa07}. It is defined as
\begin{equation}
CRPS(F,y)=\int_{-\infty}^{\infty}(F(x)-\mathbbm{1}_{\{y \leq x \}})^2dx,
\end{equation}
where $F$ is the predictive cumulative distribution, $y$ is the observed
realization, and $\mathbbm{1}$ denotes an indicator function.  If a sample $y^{(1)},\dots, y^{(m)}$ from $F$ is available, it can be
approximated by
\begin{equation}
\frac{1}{m}\sum_{i=1}^m|y^{(i)}-y|-\frac{1}{2m^2}\sum_{i,j=1}^m|y^{(i)}-y^{(j)}|.
\end{equation}

Ideally, one would run the full MCMC algorithm at each time point $t\geq 720$, including all data up to the point,
and obtain predictive distributions from this. Since this is rather time consuming, we make the following
approximation. We assume that the posterior
distribution of the ``primary'' parameters $\vect{\theta}$, $\vect{b}$, and
$\lambda$ given
$\vect{y}_{1:t}=\{\vect{y}_1,\dots,\vect{y}_{t}\}$ is the same for all
$t\geq 720$. That is, we neglect the additional information that the
observations in March provide about the primary parameters. Thus, the posterior distributions of the primary parameters are calculated
only once, namely on the data set from December 2008 to February 2009. The assumption that the posterior of the primary parameters does not change
with additional data may be questionable over longer time periods and
when one moves away from the time period from which data is used to
obtain the posterior distribution. But
since all our data lies in the winter season, we think that this assumption
is reasonable. If longer time periods are considered, one could use sliding
training windows or model the primary parameters as non-stationary using a
temporal evolution.

For each
time point $t \geq 720 $, we make up to $8$ steps ahead forecasts
corresponding to 24 hours. I.e., we sample from the
predictive distribution of $\vect{y}^*_{t+k}$, $k= 1,\dots 8$, given
$\vect{y}_{1:t}=\{\vect{y}_1,\dots,\vect{y}_{t}\}$.

 \begin{figure}[!ht]
\begin{center}
\includegraphics[width=\textwidth]{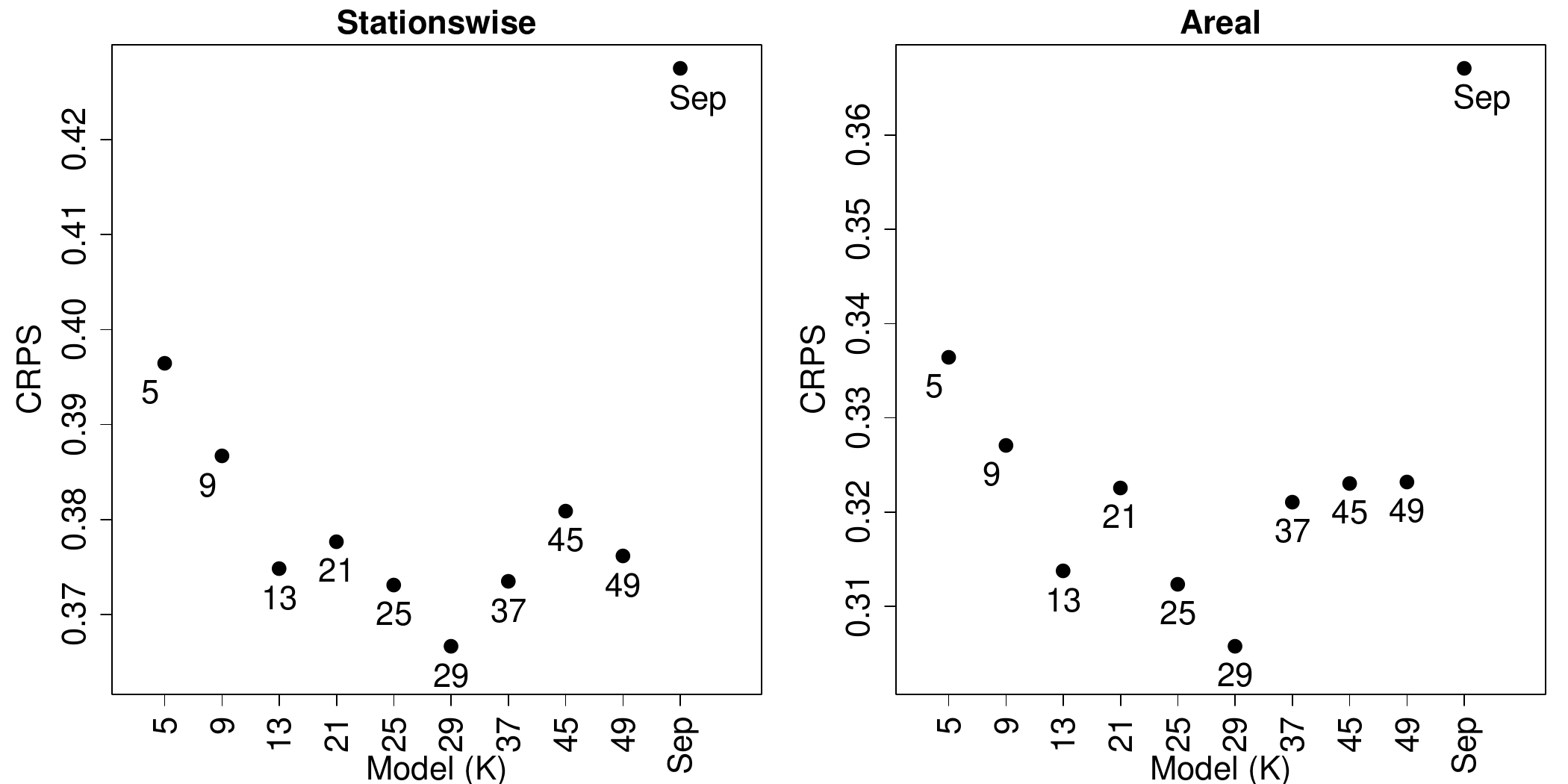} 
\end{center}
\caption{Comparison of different statistical models using the continuous ranked probability score (CRPS). On the left are CRPSs of station specific
forecasts and on the right are CRPSs of areal forecasts. $K$ denotes the
number of basis functions used in the model.
``Sep'' denotes the separable model with $K=29$ Fourier terms. The unit of the CRPS is mm.} 
\label{fig:CRPS}
\end{figure}
In Figure \ref{fig:CRPS}, the average CRPS of the pointwise predictions and
the areal predictions are shown for the different statistical models. In
the left plot, the mean is taken over all stations and lead times, whereas
the areal version is an average over all lead times. This is done for the models with different numbers of basis
functions used. Models
including only a few low-frequency Fourier terms perform worse. Then the
CRPS decreases successively. The model based on including $K=29$ Fourier
functions performs best. After this, adding higher
frequencies results into lower predictive performance. We interpret this
results in the way that the observation data does not allow for resolving high
frequencies in the error term between the forecasted and observed precipitation. Note that high frequencies of the precipitation process itself are accounted for by the forecast $y_F$. For comparison, we also fit a separable model
which is obtained by setting $\vect{\mu}=\vect 0$ and
$\mat{\Sigma}^{-1}=\mat 0_{2,2}$. Concerning the number of Fourier functions, we
use $K=29$ different Fourier terms. The separable model clearly performs worse than the
model with a non-separable covariance structure. Based on these findings, we decided
to use the model with $29$ cosine and sine functions. 

\begin{table}
\caption{Posterior medians and 95 $\%$ credible intervals for the SPDE
  based spatio-temporal model presented in Section \ref{SpecSpace} with $K=29$ Fourier terms.\label{tab:PostQuant}}
\centering
\begin{tabular}{rrrr}
  \hline
\hline
 & Median & 2.5 \% & 97.5 \% \\ 
  \hline
$\rho_0$ & 25.4 & 18.8 & 32.4 \\ 
  $\sigma^2$ & 0.838 & 0.727 & 0.994 \\ 
  $\zeta$ & 0.00655 & 0.000395 & 0.0156 \\ 
  $\rho_1$ & 48.8 & 42.1 & 57.1 \\ 
  $\gamma$ & 4.33 & 3.34 & 6.01 \\ 
  $\psi$ & 0.557 & 0.49 & 0.617 \\ 
  $\mu_x$ & 6.73 & 0.688 & 12.9 \\ 
  $\mu_y$ & -4.19 & -8.55 & -0.435 \\ 
  $\tau^2$ & 0.307 & 0.288 & 0.327 \\ 
  $b_1$ & 0.448 & 0.414 & 0.481 \\ 
  $b_2$ & -0.422 & -0.5 & -0.344 \\ 
  $\lambda$ & 1.67 & 1.64 &  1.7 \\ 
   \hline
\end{tabular}
\end{table}
Table \ref{tab:PostQuant} shows posterior medians as well as $95 \%$ credible
intervals for the different parameters. Note that the range parameters $\rho_0$
and $\rho_1$ as well as the drift parameters $\mu_x$ and $\mu_y$ have been
transformed back from the unit $[0,1]$ scale to the original km scale. The posterior median of the variance $\sigma^2$ of the innovations of the
spatio-temporal process is around $0.8$. Compared to this, the nugget
variance being about $0.3$ is smaller. For the
innovation range parameter $\rho_0$, we obtain a value of about $25$ km.  And the range parameter $\rho_1$ that controls the
amount of diffusion or, in other words, the amount of spatio-temporal
interaction, is approximately $49$ km.  With
$\gamma$ and $\psi$ being around $4$ and $0.6$, respectively, we
observe anisotropy in the south-west to north-east direction. This
is in line with the orography of the region, as the majority of the grid points lies
between two mountain ranges: the Jura to the north-west and the Alps to the
south-east. The drift points to the south-east, both parameters being
rather small though. Further, the
damping parameter $\zeta$ has a posterior median of about $0.01$.
 
\begin{table}
\caption{Comparison of NWP model and statistically postprocessed forecasts ('Stat PP') using the mean absolute error (MAE).  'Static' denotes the constant forecast obtained by using the most recently observed data. The unit of the MAE is mm.\label{CosmoMAE}}
\centering
\begin{tabular}{rrrr}
  \hline
\hline
 & Stat PP & NWP & Static \\ 
  \hline
Stationwise & 0.359 & 0.485 & 0.594 \\ 
  Areal & 0.303 & 0.387 & 0.489 \\ 
   \hline
\end{tabular}
\end{table}

Next, we compare the performance of the postprocessed forecasts with the
ones from the NWP model. In addition to the temporal cross-validation, we do
the following cross-validation in space and time. We first remove six randomly selected
stations from the data, fit the latent process to the remaining stations, and evaluate the forecasts at the stations left out. Concerning the primary parameters, i.e., all parameters except
the latent process, we use the posterior obtained from the full
data including all stations. This is done for computational simplicity and since this posterior is
not very sensitive when excluding a few stations (results not reported). Since the NWP produces 8 step ahead predictions once a day, we only consider
statistical forecasts starting at 0:00UTC. This is in contrast to
the above comparison of the different statistical models for which 8 step
ahead predictions were made at all time points and not just once for each day. We use the mean
absolute error (MAE) for evaluating the NWP forecasts. In order to be
consistent, we also generate point forecasts from the statistical
predictive distributions by using medians, and then calculate the MAE
for these point forecasts. In Table \ref{CosmoMAE}, the
results are reported. For comparison, we also give the score for the static forecast
that is obtained by using the most recently observed data. The
postprocessed forecasts clearly perform better than the raw NWP
forecasts. In addition, the postprocessed forecasts have the advantage that
they provide probabilistic forecasts quantifying prediction
uncertainty.

 \begin{figure}[!ht]
\begin{center}
\subfigure[NWP]{
\includegraphics[width=0.45\textwidth]{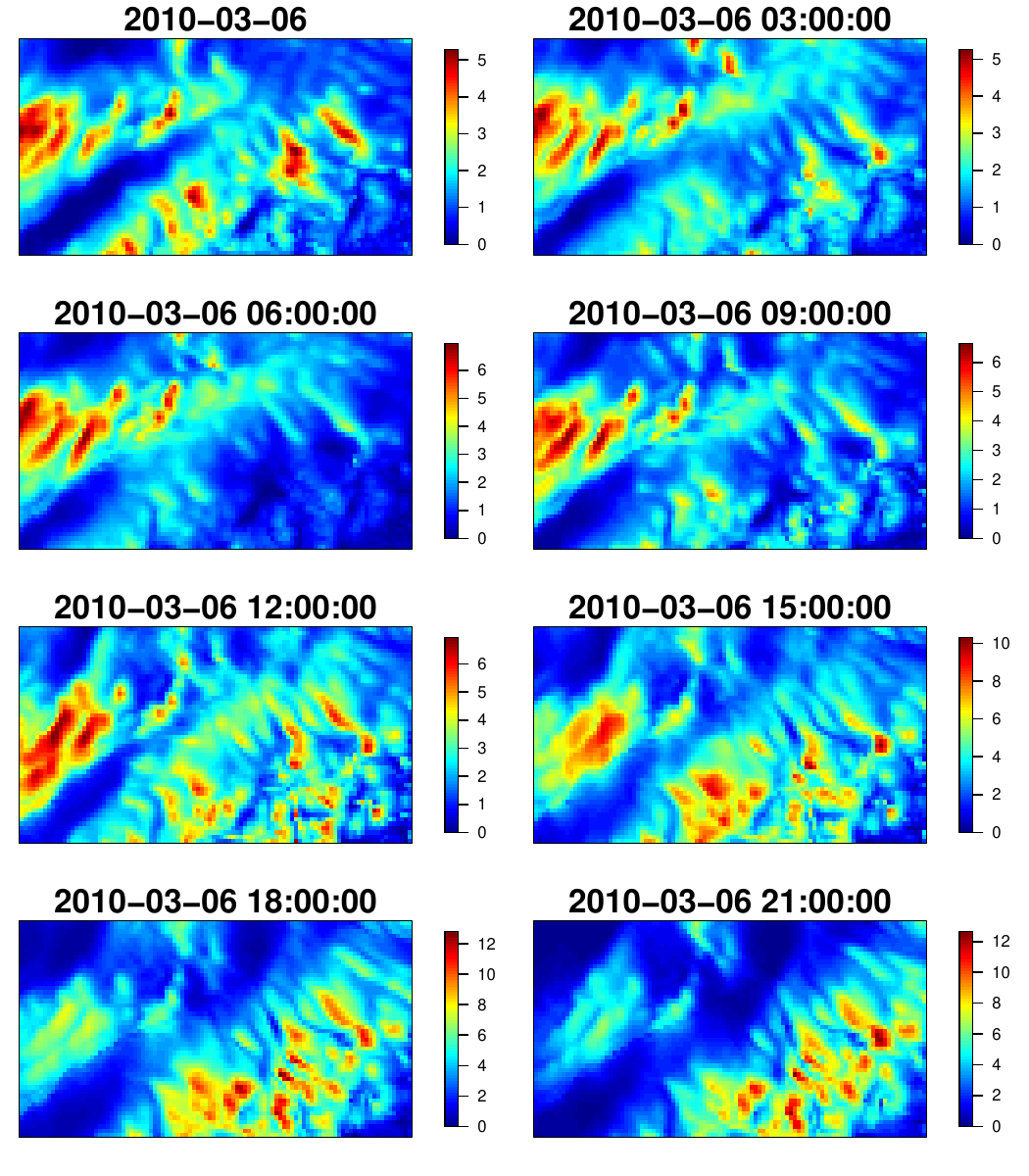} 
}
\subfigure[Median]{
\includegraphics[width=0.45\textwidth]{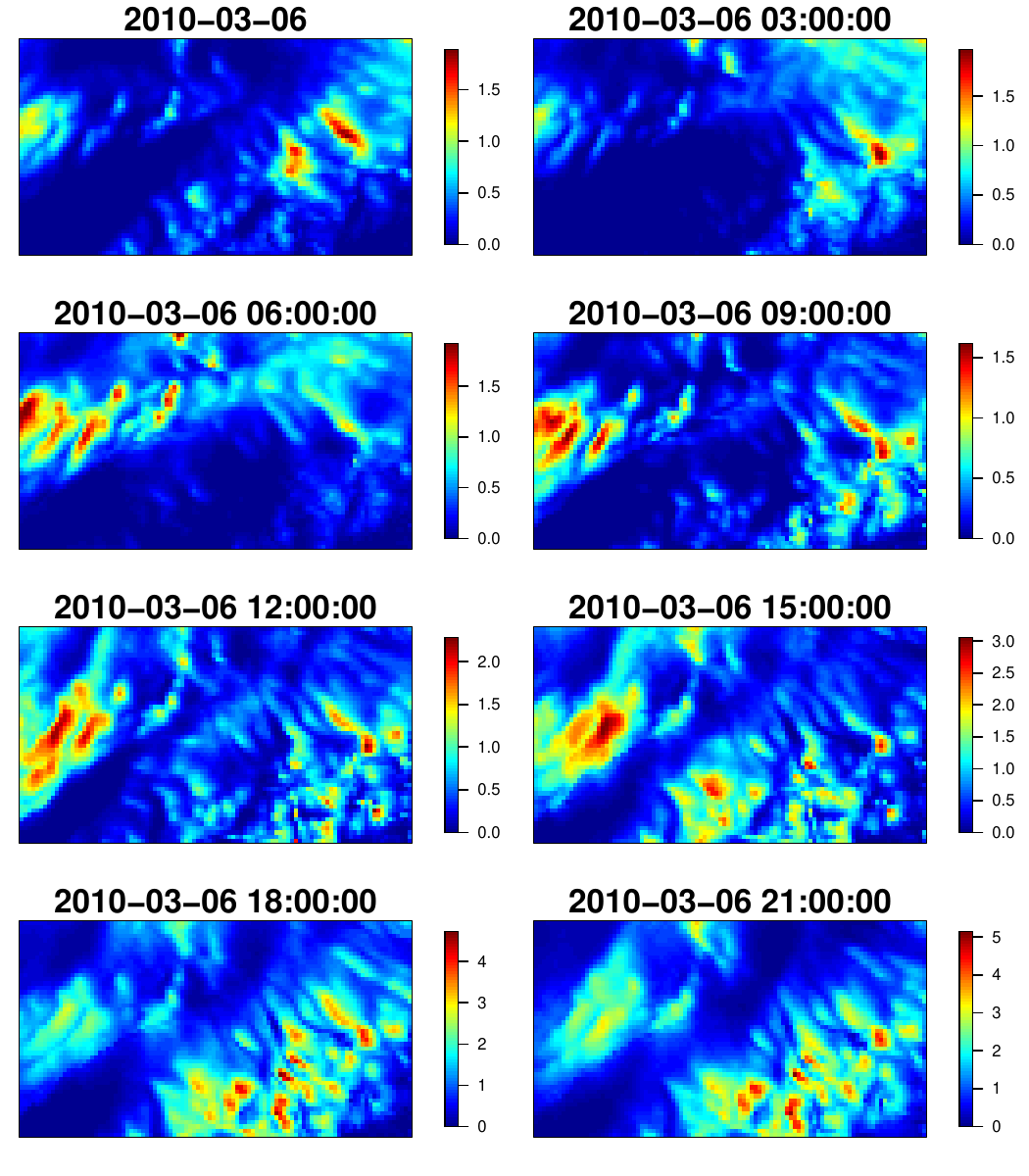} 
}
\subfigure[One sample]{
\includegraphics[width=0.45\textwidth]{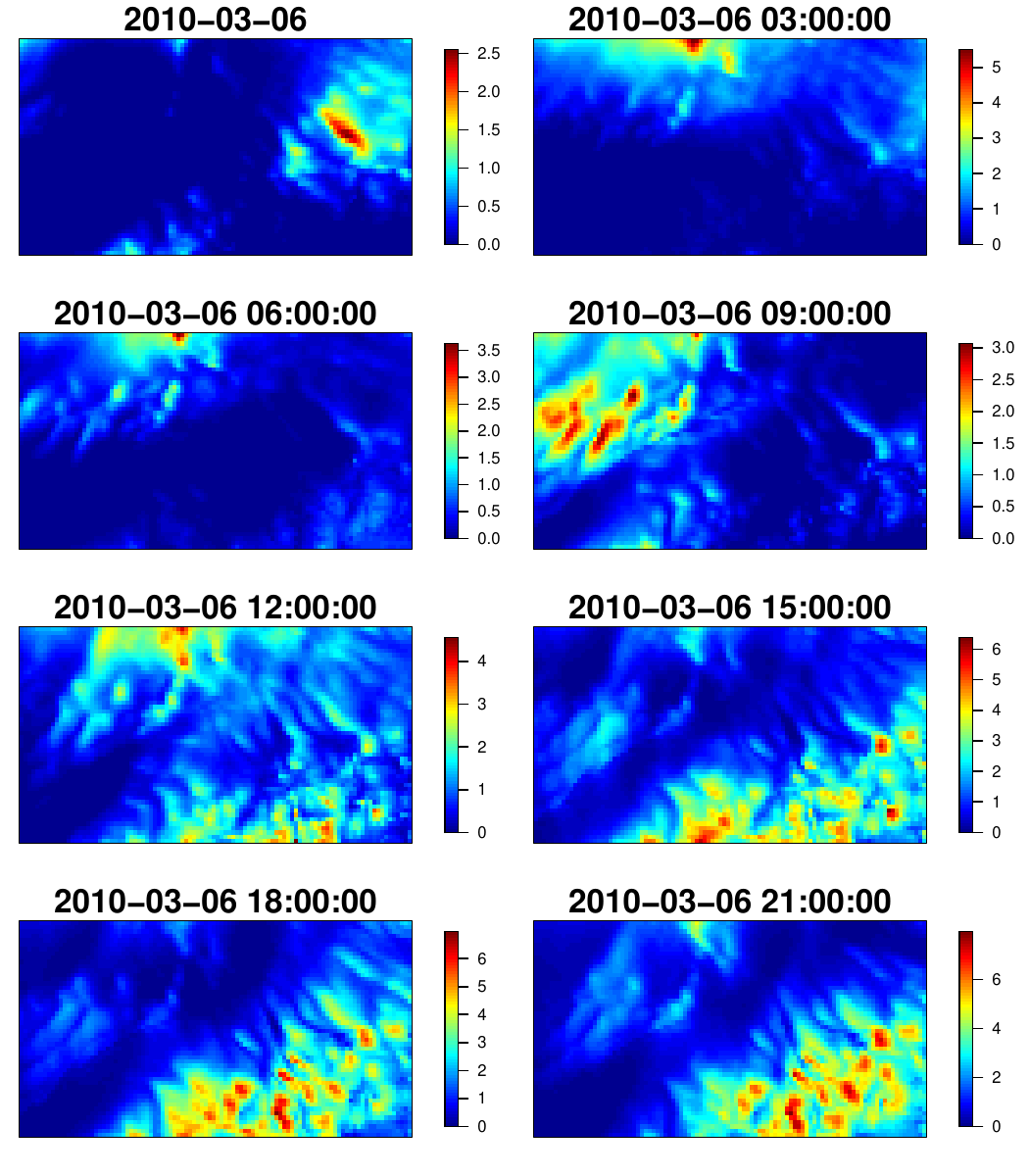} 
}
\subfigure[Quartile difference]{
\includegraphics[width=0.45\textwidth]{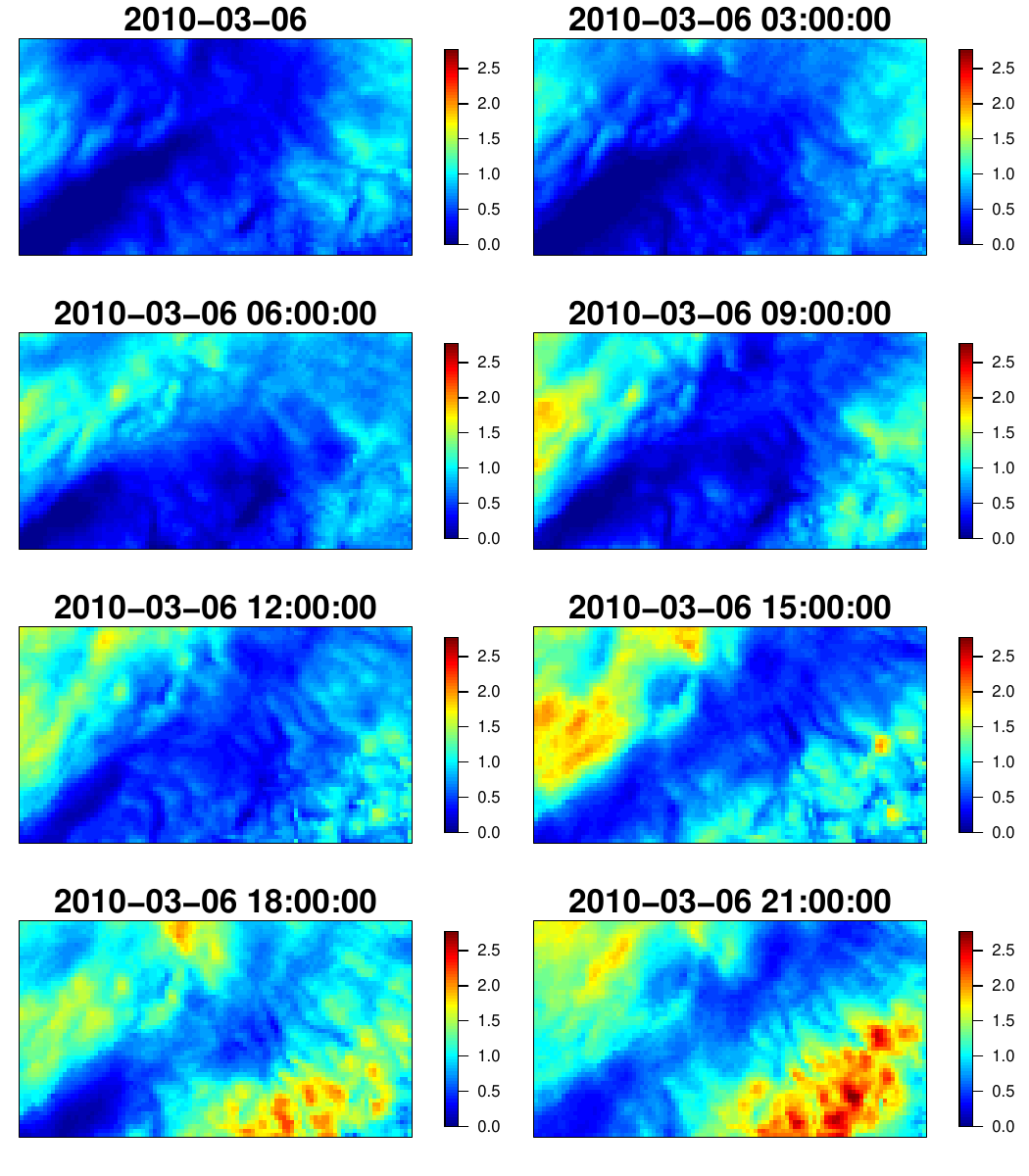} 
}
\end{center}
\caption{Illustration of postprocessed spatio-temporal precipitation
  fields for the period $t=761,\dots,768$. The figure shows the NWP forecasts (a), pointwise medians of the predictive
  distribution (b), one sample from the
  predictive distribution (c), and the differences between the third quartile and the
  median of the predictive distribution (d). All quantities are in mm. Note
that the scales are different in different figures.} 
\label{fig:Pred_Illus}
\end{figure}

The statistical model produces a joint spatio-temporal predictive
distribution that is spatially highly resolved. To illustrate the use of
the model, we show several quantities in Figure
\ref{fig:Pred_Illus}. We consider the time point $t=760$
and calculate predictive distributions over the next 24 hours. Predicted fields for the period
$t=761,\dots,768$ from the NWP are shown in the top left corner. On the right of it are pointwise medians obtained
from the statistical forecasts. This is a period during which the NWP
predicts too much rainfall compared to the observed data (results not shown). The figure shows how the statistical model
corrects for this. For illustration, we also show one sample
from the predictive distribution. To quantify prediction uncertainty, the difference between the third
quartile and the median of the predictive distribution is plotted. These
plots again show the growing uncertainty with increasing lead
time. Other quantities of interest (not shown here),
that can be easily obtained, include probabilities of
precipitation occurrence or various quantiles of the distribution.

\section{Conclusion}
We present a spatio-temporal model and corresponding efficient
algorithms for doing statistical inference for large data sets. Instead of using the covariance function, we propose
to use a Gaussian process defined trough an SPDE. The SPDE is solved using
Fourier functions, and we have given a bound on the precision of the
approximate solution. In the spectral space, one can use computationally efficient statistical algorithms whose computational costs grow linearly with
the dimension, the total computational costs being dominated by the fast
Fourier transform. The space-time Gaussian process defined through the advection-diffusion SPDE
has a nonseparable covariance structure and can be physically motivated. The model is applied to postprocessing of precipitation forecasts for northern
Switzerland. The postprocessed forecasts clearly outperform the raw NWP
predictions. In addition, they have the advantage that they quantify
prediction uncertainty.

In our analysis, we considered cumulative rainfall over 3 hours,
both in the NWP forecasts and in the station data. It would be
interesting to formulate a model which can describe different
accumulation periods in a coherent way and is still computationally
feasible.  Another interesting direction for further research would be to extend the SPDE
based model to allow for spatial non-stationarity. For instance, the deformation method
of \citet{SaGu92}, where the process is assumed to be stationary in a transformed space
and non-stationary in the original domain, might be a potential
way. Since the operators of the SPDE are local, one can define
the SPDE on general manifolds and, in particular, on the sphere (see, e.g., \cite{LiLiRu10}). Future research will show to which extent spectral methods can still be used in practice.

\section*{Acknowledgments}
We are grateful to Vanessa Stauch from MeteoSchweiz for providing the data and for
inspiring discussions. In addition, we would like to thank Peter Guttorp for
interesting comments and discussions and two referees for helpful
comments and suggestions.

\bibliographystyle{chicago}

\bibliography{/Users/fabiosigrist/Documents/Work/R/Latex/FSbiblio}

\end{document}